\documentclass{article}%
\usepackage{amsmath,amsthm}%
\usepackage{amsfonts, color}%
\usepackage{amssymb}%
\usepackage{graphicx}
\usepackage{verbatim}
\usepackage[T1]{fontenc}
\usepackage[utf8]{inputenc}
\usepackage{authblk}
\usepackage{booktabs}
\usepackage{setspace}
\usepackage{sectsty}
\usepackage{afterpage}
\usepackage{mathrsfs}

\newtheorem{lem}{Lemma}

\newcommand {\be} {\begin{equation}}
\newcommand {\ee} {\end{equation}}
\newcommand {\by} {\begin{eqnarray*}}
\newcommand {\ey} {\end{eqnarray*}}
\newcommand {\byn} {\begin{eqnarray}}
\newcommand {\eyn} {\end{eqnarray}}
\newcommand {\bd} {\begin{displaymath}}
\newcommand {\ed} {\end{displaymath}}

\RequirePackage[numbers]{natbib}

\newtheorem{theorem}{Theorem}

\newtheorem{corollary}{Corollary}

\newtheorem{definition}{Definition}

\newtheorem{remark}{Remark}

\usepackage[margin=1.00in]{geometry}
\newcommand{\bdsm}{\boldsymbol}
\begin{document}

\title{An interval-valued GARCH model for range-measured return processes}
\author[1]{Yan Sun}
\author[2]{Guanghua Lian\thanks{Dr. Guanghua Lian did most of his work on this paper while he was a Senior Lecturer in the School of Commerce at University of South Australia, Australia.}}
\author[3]{Zudi Lu}
\author[4]{Jennifer Loveland}
\author[5]{Isaac Blackhurst}
\affil[1,4,5]{Department of Mathematics $\&$ Statistics,
Utah State University,
3900 Old Main Hill,
Logan, Utah 84322,
USA.}
\affil[2]{Morgan Stanley,
20 Bank Street, London,
E14 4AD, UK
}
\affil[3]{Mathematical Sciences \& S3RI-Southampton Statistical Sciences Research Institute,
University of Southampton,

SO17 1BJ, UK}

\date{}
\maketitle

\begin{abstract}
Range-measured return contains more information than the traditional scalar-valued return. In this paper, we propose to model the [low, high] price range as a random interval and suggest an interval-valued GARCH (Int-GARCH) model for the corresponding range-measured return process. Under the general framework of random sets, the model properties are investigated. Parameters are estimated by the maximum likelihood method, and the asymptotic properties are established. Empirical application to stocks and financial indices data sets suggests that our Int-GARCH model overall outperforms the traditional GARCH for both in-sample estimation and out-of-sample prediction of volatility.
\end{abstract}


\section{Introduction}
The issue of assets' volatility plays an essential role in modern finance. It provides a measured variability for the asset price over a certain period of time, and is a key parameter in many financial applications such as financial derivatives pricing, risk assessment, and portfolio management. The squared return of log prices, being the classical estimator of variance, has often been used as the ``ideal'' proxy of volatility. As a result, many volatility models built on returns were proposed and, for a long time, have been very popular and successful. The celebrated ARCH (\cite{Engle82}) and GARCH (\cite{Bollerslev86}) models are examples of this type, which use information of traditional return series. Recently, as the high-frequency transaction data become widely available, price changes can be practically monitored in a continuous way, and the traditional low-frequency (e.g., daily) return is no longer quite representative of the volatility. For example, a small low-frequency return does not necessarily imply low volatility, as the price may fluctuate a lot and close at a similar level to the opening. On the other hand, a big return could only be the result of a very different opening price from the previous day's closing. Therefore, the traditional return-based models, using lesser information, are likely to produce inefficient or even incorrect estimates of the volatility.

In fact, since closing price is only a ``snapshot'' among numerous prices during a day, there is nothing ``special'' about it and return need not be calculated solely based on it. With the wide availability of high-frequency data, more information can obviously be utilized. For example, the difference between any two observed (log) prices in two consecutive days can be viewed as a proxy of volatility. This idea leads us to propose a naturally generalized concept of daily return, which is an interval that includes all the ``snapshot'' returns. Let $y_t\left(s\right)$ be the log price of an asset at time $s$ on day $t$. Ideally $s$ should be a continuous time index. But since price can only be observed at discrete times even for high-frequency data, $s$ is assumed to be a discrete index here. We define the interval-valued daily return as
\begin{eqnarray}
  r_t&=&\left[\min_{s,w}\left\{y_t\left(s\right)-y_{t-1}\left(w\right)\right\}, \max_{s,w}\left\{y_t\left(s\right)-y_{t-1}\left(w\right)\right\}\right]\nonumber\\
  &=&\left[\min_{s}\left\{y_t\left(s\right)\right\}-\max_{w}\left\{y_{t-1}\left(w\right)\right\}, \max_{s}\left\{y_t\left(s\right)\right\}-\min_{w}\left\{y_{t-1}\left(w\right)\right\}\right].\label{def:rt-1}
\end{eqnarray}
Namely, $r_t$ is the range of ``snapshot'' returns during one day, which is expected to contain richer information about the daily volatility than the single closing-to-closing return. Our goal is to build a volatility model that reveals the dynamics of the entire return range $r_t$ as a whole. To this end, we propose to extend the GARCH model to allow for interval-valued return input, and the resulting interval-valued model is called the Int-GARCH model. 

The essential idea of our Int-GARCH model is to enrich the return based volatility model (i.e., GARCH) with range information. Such an idea is not new. There has been a great deal of effort in the literature on volatility modeling using price range of either low-frequency or high-frequency data (\cite{Garman80}, \cite{Parkinson80}, \cite{Rogers91}, \cite{Kunitomo92}, \cite{Alizadeh02}, \cite{Chou05}, \cite{Engle06}, \cite{Brandt06}, \cite{Christensen07}). In these methods, price range is essentially viewed as an exogenous variable and is included into the modeling via functions of certain forms. In a more systematic way, our approach integrates the level and range information by modeling the [low, high] return range $r_t$ as a whole in the framework of random interval. As is shown in Section \ref{sec:model}, based on the fitted model, our final prediction of volatility is made such that all the information contained in $r_t$ is systematically accounted for. 

In addition to volatility prediction, our Int-GARCH model in a broader sense makes a pioneer contribution to the study of conditional heteroscedasticity for interval-valued time series. It has been a while since interval-valued time series was investigated and there are a handful of results mainly on the mean models, i.e., models that aim at making prediction of the interval-valued mean, and their applications (e.g.,\cite{Maia08}, \cite{Han12}). The focus of our model, by comparison, is the conditional variance. Variance model for interval-valued time series, as far as the authors are aware, largely remains an unexploited area, and our study is among the very first investigations.  

Our theoretical results are two-fold. We first establish that under certain conditions our Int-GARCH model achieves weak stationarity that is characterized by a time invariant mean and variance. Then, under the assumption of weak stationarity, we define and obtain the explicit formula of the autocorrelation function (ACF) of the Int-GARCH process. We propose to estimate the model parameters by the method of maximum likelihood and provide the associated asymptotic properties. Simulation shows that the results are consistent with our theoretical findings. For empirical study, we analyze several stocks and indices data that are representative of the market. Our Int-GARCH model is compared to GARCH using the RV as the market proxy. Based on both in-sample and out-of-sample comparisons, our Int-GARCH model overall outperforms GARCH with higher correlations and reduced errors to RV.

The rest of the paper is organized as follows. Random sets preliminaries are provided in Section \ref{sec:prelim}. We formally introduce our Int-GARCH model and its general properties in Section \ref{sec:model}. Stationarity of Int-GARCH (1,1,1) is presented in Section \ref{sec:model_111}. Section \ref{sec:mle} discusses the maximum likelihood estimator for the model parameters and carefully investigates its performances by a simulation study. Empirical study with the stocks and indices data, as well as a detailed discussion, are reported in Section \ref{sec:empirical}. We finish with concluding remarks in Section \ref{sec:conclude}. Proofs and useful lemmas are provided in the Appendix.

\section{Preliminaries of random sets}\label{sec:prelim}
Throughout the paper, we will view $r_t$ as a random interval and model its dynamics under the framework of random sets. To facilitate our presentation, we briefly introduce the basic notations and definitions in the random set theory. For more details we refer the readers to \cite{Kendall74}, \cite{Matheron75}, \cite{Artstein75}, \cite{Molchanov05}, \cite{Sun15}, among others. 

Let $(\Omega,\mathcal{L},P)$ be a probability space. Denote by $\mathcal{K}\left(\mathbb{R}^d\right)$ or $\mathcal{K}$ the collection of all non-empty compact subsets of $\mathbb{R}^d$. In the space $\mathcal{K}$, a linear structure is defined by Minkowski addition and scalar multiplication, i.e.,
\begin{equation}\label{set-lin}
  A+B=\left\{a+b: a\in A, b\in B\right\},\ \ \ \ \lambda A=\left\{\lambda a: a\in A\right\},
\end{equation}
$\forall A, B\in\mathcal{K}$ and $\lambda\in\mathbb{R}$. A random compact set is a Borel measurable function $A: \Omega\rightarrow\mathcal{K}$, $\mathcal{K}$ being equipped with the Borel $\sigma$-algebra induced by the Hausdorff metric. For each $A\in\mathcal{K}\left(\mathbb{R}^d\right)$, the function defined on the unit sphere $S^{d-1}$:
\begin{equation*}
  s_A\left(u\right)=\sup_{a\in A}\left<u, a\right>,\ \ \forall u\in S^{d-1},
\end{equation*}
is called the support function of A. If $A(\omega)$ is convex almost surely, then $A$ is called a random compact convex set. Especially, a one-dimensional random compact convex set is called a random interval. Let $\mathcal{K}_\mathcal{C}(\mathbb{R}^d)$ denote the space of non-empty compact convex subsets of $\mathbb{R}^d$. An $L_2$ metric in $\mathcal{K}_\mathcal{C}(\mathbb{R}^d)$ is given via the support function by
\begin{equation*}
  \rho_2(A, B)=\left[d\int_{S^{d-1}}|s_A(u)-s_B(u)|^2\mu(du)\right]^{\frac{1}{2}}. 
\end{equation*}
Much of the random sets theory has focused on compact convex sets via their support functions. For a random compact convex set $X$, the well accepted expectation is defined by the Aumann integral of set-valued function (\cite{Aumann65}) as
\begin{equation}
  \text{E}(X)=\left\{\text{E}\xi: \xi\in X\ a.s., \text{E}\left\|\xi\right\|<\infty\right\}.
\end{equation}
Alternatively, \cite{Frechet48} gave a general definition for the expectation of a random element $X$ in the metric space $(\mathcal{K}_\mathcal{C}(\mathbb{R}^d), \rho)$ as the solution of 
\begin{equation}
  \text{E}\rho^2(X, \text{E}_F(X))=\inf_{A\in\mathcal{K}_\mathcal{C}} \text{E}\rho^2(X, A).
\end{equation}
If we choose the metric $\rho$ to be $\rho_2$, then the Fr\'echet expectation $\text{E}_F(X)$ coincides with the Aumann expectation $\text{E}(X)$, and the variance is further defined as 
\begin{equation*}
    \text{Var}(X)=\text{E}\rho_2^2(X, \text{E}(X)). 
\end{equation*}
See \cite{Lyashenko82}, \cite{Korner95}, and \cite{Korner97}. Restricting to $\mathcal{K}_\mathcal{C}(\mathbb{R})$, the Aumann expectation of a random interval $X$ is simply   
\begin{equation}\label{set-mean}
  \text{E}(X)=\left[\text{E}(X^C)-\text{E}(X^R), \text{E}(X^C)+\text{E}(X^R)\right],
\end{equation}
where $(\cdot)^C$ and $(\cdot)^R$ denote the corresponding center and radius, respectively. The $\rho_2$ distance between two intervals $x$ and $y$ is (\cite{Korner95})
\begin{equation}\label{L2-metric}
  \rho_2(x, y)=\left[(x^C-y^C)^2+(x^R-y^R)^2\right]^{\frac{1}{2}}.
\end{equation}
Thus, the variance of a random interval is consequently calculated as
\begin{eqnarray}\label{set-var}
    \text{Var}(X)=\text{E}\left[(X^C-\text{E}(X^C))^2+(X^R-\text{E}(X^R))^2\right]
    =\text{Var}(X^C)+\text{Var}(X^R).  
\end{eqnarray}
Let $\mathcal{F}$ be any $\sigma$-filed on $\Omega$. The conditional expectation of a random set given $\mathcal{F}$ is defined
by the integral and conditional expectation of multivalued functions (\cite{Hiai77}). Specifically for a random interval $X$,
\begin{equation}\label{def:cond-exp}
  \text{E}(X|\mathcal{F})=\left[\text{E}(X^C|\mathcal{F})-\text{E}(X^R|\mathcal{F}), \text{E}(X^C|\mathcal{F})+\text{E}(X^R|\mathcal{F})\right]. 
\end{equation}
Based on the conditional expectation, the conditional variance of a random set is defined as the conditional mean 
squared distance from its conditional expectation (\cite{Nather07}). According to the definition, the conditional variance of $X$ with respect
to the distance $\rho_2$ is given by 
\begin{eqnarray}
  \text{Var}(X|\mathcal{F})
  &=&E\left[\rho_2^2\left(X, E\left(X|\mathcal{F}\right)\right)|\mathcal{F}\right]\nonumber\\
  &=&E\left[\left(X^C-E(X|\mathcal{F})^C\right)^2+\left(X^R-E(X|\mathcal{F})^R\right)^2|\mathcal{F}\right]\nonumber\\
  &=&E\left[\left(X^C-E(X^C|\mathcal{F})\right)^2+\left(X^R-E(X^R|\mathcal{F})\right)^2|\mathcal{F}\right]\nonumber\\
  &=&E\left[\left(X^C-E(X^C|\mathcal{F})\right)^2|\mathcal{F}\right]
  +E\left[\left(X^R-E(X^R|\mathcal{F})\right)^2|\mathcal{F}\right]\nonumber\\
  &=&\text{Var}(X^C|\mathcal{F})+\text{Var}(X^R|\mathcal{F}).\label{def:cond-var} 
\end{eqnarray}

\section{The Int-GARCH Model}\label{sec:model}
\subsection{Model specification}
We assume observing a range-measured return series $\left\{r_t\right\}_{t=1}^{T}$ of the form
$$r_t=[\lambda_t-\delta_t, \lambda_t+\delta_t],\ t=1,2,\cdots,T.$$
That is, $\left\{\lambda_t\right\}_{t=1}^{T}$ and $\left\{\delta_t\right\}_{t=1}^{T}$ are the associated center and radius processes, both of which are observable. According to (\ref{set-mean}) and (\ref{set-var}), the mean and variance of $r_t$ as a random interval are
\begin{eqnarray}
  \text{E}(r_t)&=&\left[\text{E}(\lambda_t)-\text{E}(\delta_t), \text{E}(\lambda_t)+\text{E}(\delta_t)\right],\label{meanA}\\
  \text{Var}(r_t)&=&\text{Var}(\lambda_t)+\text{Var}(\delta_t)\label{varA}.
\end{eqnarray}

Let $\mathcal{F}_t$ denote the information set up to time t, i.e. $\mathcal{F}_t=\sigma\left\{r_s: s\leq t\right\}.$ We are concerned with the conditional variance $H_t^2$ of $r_t$ given $\mathcal{F}_{t-1}$. According to (\ref{def:cond-var}), $H_t^2$ with respect to $\rho_2$ is computed as 
\begin{equation*}
  H_t^2=\text{Var}(r_t|\mathcal{F}_{t-1})
  =\text{Var}(\lambda_t|\mathcal{F}_{t-1})+\text{Var}(\delta_t|\mathcal{F}_{t-1}).
\end{equation*}
The GARCH model depicts the conditional variance of a point-valued return process as a linear function of the past squared returns and variances. This was inspired by the fact that assets returns usually exhibit volatility clustering: large variations in prices tend to cluster together, resulting in separate dynamic and tranquil periods of the market. Extending this spirit to the interval-valued process $\left\{r_t\right\}$, one would expect, conceptually, a model like
\begin{equation}\label{model}
  H_t^2=g\left(H_s^2, \hat{H}_s^2(r_s): s\leq t-1\right),
\end{equation}
where $\hat{H}_s^2(r_s)$ denotes a return range based proxy for $H_s^2$, and $g$ is linear in $H_s^2$ and $\hat{H}_s^2(r_s)$. Assuming $\text{E}(\lambda_t|\mathcal{F}_{t-1})=0$ and $\text{E}(\delta_t|\mathcal{F}_{t-1})=c>0$, a reasonable function $g$ in (\ref{model}) seems to imply
\begin{equation*}
  H_t^2=\mu+\sum_{i=1}^{p}\alpha_i\left[\lambda_{t-i}^2+\delta_{t-i}^2-c^2\right]
  +\sum_{i=1}^{q}\beta_iH_{t-i}^2,
\end{equation*}
where $p>0, q\geq 0$. The constant $c$ can be absorbed into the parameter $\mu$ and the above equation is simplified to
\begin{equation*}
  H_t^2=\mu+\sum_{i=1}^{p}\alpha_i\left[\lambda_{t-i}^2+\delta_{t-i}^2\right]+\sum_{i=1}^{q}\beta_iH_{t-i}^2.
\end{equation*}
To give more flexibility to our model, we allow for different degrees of dependence of $H_t$ on the past centers and radii. In addition, adopting the idea of more robust modeling of volatility by \cite{Taylor86} and \cite{Schwert90}, we propose to model the conditional standard deviation $H_t$ directly, instead of via the conditional variance $H_t^2$ (\cite{Ding93} also considered such a specification as a special case of their A-PARCH model).

Given the above discussion, our Int-GARCH $(p,q,w)$ model for the return range process is specified as
\begin{eqnarray}
  r_t&=&h_t\cdot v_t,\label{igarch_1}\\
  v_t&=&[\epsilon_t-\eta_t, \epsilon_t+\eta_t],\label{igarch_2}\\
  \epsilon_t&\stackrel{i.i.d.}{\sim}&N(0,1),\label{igarch_3}\\
  \eta_t&\stackrel{i.i.d.}{\sim}&\Gamma(k,1),\label{igarch_4}\\
  h_t&=&\mu+\sum_{i=1}^{p}\alpha_i|\lambda_{t-i}|+\sum_{i=1}^{q}\beta_i\delta_{t-i}+\sum_{i=1}^{w}\gamma_ih_{t-i},\label{igarch_5}
\end{eqnarray}
where $p>0, q>0, w\geq0$, and $\left\{\alpha_i: i=1,\cdots,p\right\}, \left\{\beta_i: i=1,\cdots q\right\}, \left\{\gamma_i: i=1,\cdots,w\right\}$ are positive constants. In addition, the error terms $\epsilon_t$ and $\eta_t$ are assumed to be independent. In (\ref{igarch_1}), ``$\cdot$'' denotes the scalar multiplication. Under this specification, the conditional variance of $r_t$ is seen to be
\begin{eqnarray}\label{H2}
  H_t^2=\text{Var}(h_t\epsilon_t|\mathcal{F}_{t-1})+\text{Var}(h_t\eta_t|\mathcal{F}_{t-1})=h_t^2(1+k).
\end{eqnarray}

Although we impose parametric assumptions on the random errors $\epsilon_t$ and $\eta_t$ to simplify our presentation here, they are not really necessary. In practice, it is best to use the true data generating distributions, which vary from data to data. So, in replacement of (\ref{igarch_3})-(\ref{igarch_4}), a relaxed yet sufficient specification for $\epsilon_t$ and $\eta_t>0$ is
\begin{eqnarray*}
  &&\text{Var}\left(\epsilon_t|\mathcal{F}_{t-1}^{\epsilon, \eta}\right)=1,\\
  &&\text{Var}\left(\eta_t|\mathcal{F}_{t-1}^{\epsilon, \eta}\right)=k.\\
\end{eqnarray*}

\begin{remark}
A more general metric for $\mathcal{K}_{\mathcal{C}}(\mathbb{R})$ was proposed by \cite{Gil01}, which essentially takes the form 
\begin{equation}\label{def:w2}
   \rho_W^2\left(x,y\right)
   =\left(x^C-y^C\right)^2+\left(x^R-y^R\right)^2\int_{[0,1]}\left(2\lambda-1\right)^2dW(\lambda),\ \ 
   x,y\in\mathcal{K}_{\mathcal{C}}(\mathbb{R}),
\end{equation} 
where $W$ is any non-degenerate symmetric measure on $[0, 1]$. Compared to the $\rho_2$ metric in (\ref{L2-metric}), the flexibility of $\rho_W$ lies in its choice of a weight between the center and radius. Since we use different parameters $\alpha$'s and $\beta$'s for the center and radius, respectively, this flexibility is in fact accounted for in our specification of $h_t$ in (\ref{igarch_5}).
\end{remark}

\subsection{Volatility forecasting}
The $h_t$ in the model (\ref{igarch_1})-(\ref{igarch_5}) relates to the conditional standard deviation $H_t$ of the return range $r_t$ according to equation (\ref{H2}). However, $H_t$ is not exactly the daily volatility as in the literature. To see the relation between $h_t$ and the daily volatility $\sigma^2_t$, we notice that $r_t$ contains all possible returns of day $t$, using different prices during days $t$ and $t-1$. Consider an arbitrary return at position $\omega$ in $r_t$, denoted by $r_t(\omega)$, as the volatility proxy, i.e.
\begin{equation*}
  r_t(\omega)=h_t\left(\epsilon_t+\omega\eta_t\right),\ \omega\in[-1,1].
\end{equation*}
Volatility based on $r_t(\omega)$ is calculated as
\begin{equation*}
  \sigma^2_t\left(\omega\right)=\text{Var}\left(r_t(\omega)|\mathcal{F}_{t-1}\right)=\left(1+\omega^2k\right)h_t^2.
\end{equation*}
Our Int-GARCH volatility $\sigma^2_t$ is defined as the average of $\left\{\sigma^2_t(\omega): \omega\in [-1,1]\right\}$. Assuming equal weight for each point return, $\sigma^2_t$ can be calculated as
\begin{eqnarray}\label{def:int-vol}
  \sigma^2_t=\frac{\int_{-1}^{1}\sigma^2_t(\omega)d(\omega)}{\int_{-1}^{1}d(\omega)}
  =\frac{1}{2}\int_{-1}^{1}\left(1+\omega^2k\right)h_t^2d(\omega)=\left(1+\frac{1}{3}k\right)h_t^2.
\end{eqnarray}
Therefore, volatility forecast is essentially made by predicting $h_t$. The 1-step-ahead prediction of $h_t$ is immediately defined by equation (\ref{igarch_5}) as
\begin{equation}\label{1-step-pred}
  \hat{h}_t\left(1\right)=\mu+\sum_{i=1}^{p}\alpha_i|\lambda_{t+1-i}|+\sum_{i=1}^{q}\beta_i\delta_{t+1-i}+\sum_{i=1}^{w}\gamma_ih_{t+1-i}.
\end{equation}
To make a 2-step-ahead prediction, notice that
\begin{eqnarray*}
  h_{t+2}&=&\mu+\sum_{i=1}^{p}\alpha_i|\lambda_{t+2-i}|+\sum_{i=1}^{q}\beta_i\delta_{t+2-i}+\sum_{i=1}^{w}\gamma_ih_{t+2-i}\\
  &=&\mu+\left(\alpha_1|\epsilon_{t+1}|+\beta_1\eta_{t+1}+\gamma_1\right)h_{t+1}
  +\sum_{i=2}^{p}\alpha_i|\lambda_{t+2-i}|+\sum_{i=2}^{q}\beta_i\delta_{t+2-i}+\sum_{i=1}^{w}\gamma_ih_{t+2-i}.
\end{eqnarray*}
Replacing $|\epsilon_{t+1}|$ and $\eta_{t+1}$ with their expectations, and $h_{t+1}$ with $\hat{h}_{t}\left(1\right)$, we get
\begin{equation}\label{2-step-pred}
  \hat{h}_t\left(2\right)=\mu+\left(\alpha_1\sqrt{\frac{2}{\pi}}+\beta_1k+\gamma_1\right)\hat{h}_{t}\left(1\right)
  +\sum_{i=2}^{p}\alpha_i|\lambda_{t+2-i}|+\sum_{i=2}^{q}\beta_i\delta_{t+2-i}+\sum_{i=1}^{w}\gamma_ih_{t+2-i}.
\end{equation}
The general $l$-step-ahead prediction can be calculated recursively. In particular, the formula for Int-GARCH (1,1,1) model is given by
\begin{equation}\label{l-step-pred}
   \hat{h}_t\left(l\right)=\mu+\left(\alpha_1\sqrt{\frac{2}{\pi}}+\beta_1k+\gamma_1\right)\hat{h}_{t}\left(l-1\right),\ l>1.
\end{equation}

\subsection{Mean stationarity}
We provide the necessary and sufficient conditions of mean stationarity for the Int-GARCH (p,q,w) model in the following theorem.
\begin{theorem}\label{thm:mean-gen}
Consider the general Int-GARCH model (\ref{igarch_1})-(\ref{igarch_5}). Define
\begin{equation}
  x_{i,t}=\alpha_i|\epsilon_t|I_{\left\{1\leq i\leq p\right\}}+\beta_i|\eta_t|I_{\left\{1\leq i\leq q\right\}}+\gamma_iI_{\left\{1\leq i\leq w\right\}},\label{def:X_it}
\end{equation}
and
\begin{equation}
  E\left(x_{i,t}\right)=\mu_i,\label{def:mu_i}
\end{equation}
where $i=1,2,\cdots,m=\max\left\{p,q,w\right\}$. Assume $\left\{r_t\right\}$ starts from its infinite past with a finite mean. Then, $Eh_t<\infty$
if and only if $\sum_{i=1}^{k}\mu_i<1$. When this condition is satisfied,
\begin{equation}\label{mean_h}
 Eh_{t}=\frac{\mu}{1-\sum_{i=1}^{m}\mu_i},
\end{equation}
and
\begin{equation}\label{mean_e}
 Er_t=\left[-kE\left(h_{t}\right), kE\left(h_{t}\right)\right].
\end{equation}
 \end{theorem}

\section{Stationarity of Int-GARCH (1,1,1)}\label{sec:model_111}
Similar to the traditional GARCH model, the Int-GARCH(1,1,1) process is a simple but effective model for analyzing interval-valued time series with conditional heteroskedasticity. In this section, we derive several important distributional properties of Int-GARCH (1,1,1). Before we present our theoretical results, we first notice that for the Int-GARCH (1,1,1) process,
\begin{eqnarray*}
  h_t&=&\mu+\alpha_1|\lambda_{t-1}|+\beta\delta_{t-1}+\gamma_1h_{t-1}\\
  &=&\mu+\alpha_1|\epsilon_{t-1}|h_{t-1}+\beta\eta_{t-1}h_{t-1}+\gamma_1h_{t-1}\\
  &=&\mu+\left(\alpha_1|\epsilon_{t-1}|+\beta\eta_{t-1}+\gamma_1\right)h_{t-1}.
\end{eqnarray*}
Defining the i.i.d. random variables $x_t=\alpha_1|\epsilon_{t-1}|+\beta\eta_{t-1}+\gamma_1$, $t\in\mathbb{N}$, $h_t$ can be re-written as
\begin{equation}\label{def_x}
  h_t=\mu+x_th_{t-1}.
\end{equation}
We will use (\ref{def_x}) throughout this section to simplify notations.

 \subsection{Weak stationarity}
It is derived in \cite{Korner95} that the covariance between two random intervals with respect to the $\rho_2$ metric is the sum of the covariances between the two centers and two radii. This implies
\begin{equation}\label{covAB}
  \text{Cov}(r_t, r_s)=\text{Cov}(\lambda_t, \lambda_s)+\text{Cov}(\delta_t, \delta_s),\ s,t\in\mathbb{N}.
\end{equation}
We are ready to extend the notion of weak stationarity to interval-valued time series in the obvious way.
\begin{definition}
An interval-valued time series $\left\{r_t\right\}$ is said to be weakly stationary, or second-moment stationary, if its unconditional mean $\text{E}\left(r_t\right)$ and covariance
$\text{Cov}\left(r_t, r_{t+s}\right)$ exist and are independent of time $t$ for all integers $s$, where $\text{E}\left(r_t\right)$ and $\text{Cov}\left(r_t, r_{t+s}\right)$ are given in (\ref{meanA}) and (\ref{covAB}), respectively.
\end{definition}
The existences of $\text{E}\left(r_t\right)$ and $\text{Var}\left(r_t\right)$ are closely related to those of $\text{E}\left(h_t\right)$ and $\text{E}\left(h_t^2\right)$, respectively. In fact, $\text{E}(h_t^2)<\infty$ implies the existences of the first two moments of $r_t$. We give precise conditions in the following two theorems.
\begin{theorem}\label{thm:mean}
Consider the Int-GARCH model (\ref{igarch_1})-(\ref{igarch_5}) with $p=q=w=1$. Assume $\left\{r_t\right\}$ starts from its infinite past with a finite mean. Then, $\text{E}(h_t)<\infty$
if and only if $\text{E}(x_t)<1$, i.e.
$$\alpha_1\sqrt{\frac{2}{\pi}}+\beta_1k+\gamma_1<1.$$
When this condition is satisfied,
\begin{equation}\label{mean_h}
 \text{E}(h_{t})=\frac{\mu}{1-\alpha_{1}\sqrt{2/\pi}-\beta_{1}k-\gamma_{1}},
\end{equation}
and
\begin{equation}\label{mean_e}
 \text{E}(r_{t})=\left[-k\text{E}\left(h_{t}\right), k\text{E}\left(h_{t}\right)\right].
\end{equation}
 \end{theorem}

\begin{theorem}\label{thm:var}
 Consider the Int-GARCH(1,1,1) model $\left\{r_t\right\}$ as in Theorem \ref{thm:mean}. Assume $\left\{r_t\right\}$ starts from its infinite past with a finite variance. Then
 $E\left(h_{t}^{2}\right)<\infty$ if and only if $E\left(x_{t}^{2}\right)<1$, i.e.
 \begin{eqnarray*}
\alpha_{1}^{2}+\beta_{1}^{2}\left(k+k^{2}\right)+\gamma_{1}^{2}+2\alpha_{1}\beta_{1}\sqrt{\dfrac{2}{\pi}}k+2\alpha_{1}\gamma_{1}\sqrt{\dfrac{2}{\pi}}+2\beta_{1}\gamma_{1}k
<1.
 \end{eqnarray*}
 When this condition is satisfied,
 \begin{eqnarray}\label{mean_h2}
   E\left(h_{t}^{2}\right)=\mu^{2}\dfrac{C_{1}+1}{\left(C_{2}-1\right)\left(C_{1}-1\right)},
 \end{eqnarray}
 and
 \begin{eqnarray}\label{var_e}
   Var\left(r_{t}\right)=\left(1+k+k^{2}\right)E\left(h_{t}^{2}\right)-k^{2}\left[E\left(h_{t}\right)\right]^{2},
 \end{eqnarray}
 where $E\left(h_t\right)$ is given in (\ref{mean_h}), and $C_{1}=E\left(x_{t}\right)$, $C_{2}=E\left(x_{t}^{2}\right)$.
\end{theorem}
So far we have found equivalent conditions for the existence of the unconditional mean and variance. In order to achieve weak stationarity, we still need to find conditions under which the (unconditional) covariances are finite and time-invariant. According to the following Theorem \ref{thm:cov}, these conditions turn out to be the same as those for the existence of variance. This is not surprising as
\begin{eqnarray*}
  |\text{Cov}\left(r_t, r_{t+h}\right)| \leq
  |\text{Cov}\left(\lambda_t, \lambda_{t+h}\right)|+|\text{Cov}\left(\delta_t, \delta_{t+h}\right)|
  \leq \text{Var}\left(\lambda_t\right)+\text{Var}\left(\delta_t\right).
\end{eqnarray*}
We summarize this conclusion in Corollary \ref{cor:stat} following Theorem \ref{thm:cov}.

 \begin{theorem}\label{thm:cov}
Consider the Int-GARCH(1,1,1) process $\left\{r_t\right\}$. Under the assumptions of Theorem \ref{thm:var}, the covariance of any two random intervals $r_t$ and $r_{t+s}$ is given by
\begin{align*}
\mbox{Cov}\left(r_{t},r_{t+s}\right) & =\begin{cases}
\left(1+k+k^{2}\right)E\left(h_{t}^{2}\right)-k^{2}\left[E\left(h_{t}\right)\right]^{2}, & s=0;\\
kE\left(h_{t}h_{t+s}\eta_{t}\right)-k^{2}\left[E\left(h_{t}\right)\right]^{2}, & |s|>0,
\end{cases}
\end{align*}
where $E\left(h_t\right)$ and $E\left(h_t^2\right)$ are given in (\ref{mean_h}) and (\ref{mean_h2}), respectively, and $E\left(h_{t}h_{t+s}\eta_{t}\right)$ is calculated explicitly in Lemma \ref{eta-X and h-h-eta} (see Appendix).
\end{theorem}

\begin{corollary}\label{cor:stat}
The Int-GARCH(1,1,1) process is weakly stationary, or second-moment stationary, if and only if $E\left(x_{t}^{2}\right)<1$.
\end{corollary}

 \subsection{Auto-correlation function (ACF)}
 The notion of the variance and covariance for compact convex random sets were naturally extended to the correlation coefficient of two random sets $A$ and $B$, which is defined as
 \begin{equation}\label{corr}
   \text{Corr}\left(A,B\right)=\frac{\text{Cov}\left(A,B\right)}{\sqrt{\text{Var}\left(A\right)\text{Var}\left(B\right)}}.
 \end{equation}
Based on this definition, we calculate the auto-correlation function (ACF) of the Int-GARCH(1,1,1) process and give the result in the corollary below.
\begin{corollary}\label{cor:acf}
Under the assumptions of Theorem \ref{thm:var}, the auto-correlation function of the Int-GARCH(1,1,1) process
$\left\{ r_{t}\right\} $ is
\[
\rho(s)=\begin{cases}
1, & s=0\\
\dfrac{kE\left(h_{t}h_{t+s}\eta_{t}\right)-k^{2}\left[E\left(h_{t}\right)\right]^{2}}{\left(1+k+k^{2}\right)E\left(h_{t}^{2}\right)-k^{2}\left[E\left(h_{t}\right)\right]^{2}}, & |s|>0.
\end{cases}
\]
 \end{corollary}

We plot the ACF for a specific Int-GARCH(1,1,1) model (Model I in the simulation) in Figure \ref{fig:realACF_1}. We see that the centers are uncorrelated. This has been verified by (\ref{eqn2}) in the proof of Theorem \ref{thm:cov}. 
The radii, or the lengths of the intervals, have a relatively persistent auto-correlation, which coincides with the phenomenon of ``volatility clustering''. This long-term dependence of radii carries over to the intervals as a whole, and results in a slow-dying ACF of the interval-valued process. \\

\begin{figure}[ht]
\centering
\includegraphics[ height=1.800in, width=2.300in]{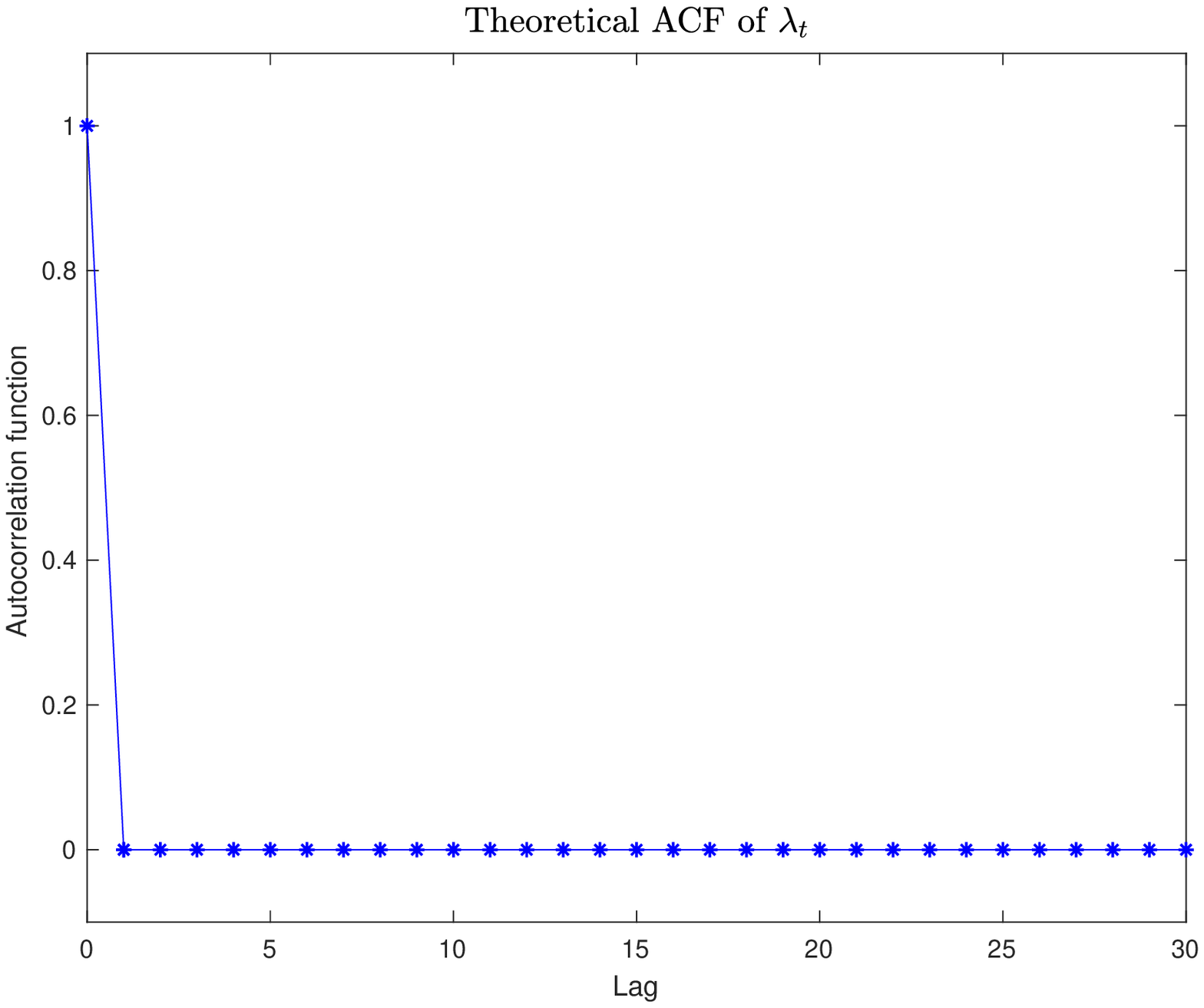}
\includegraphics[ height=1.800in, width=2.300in]{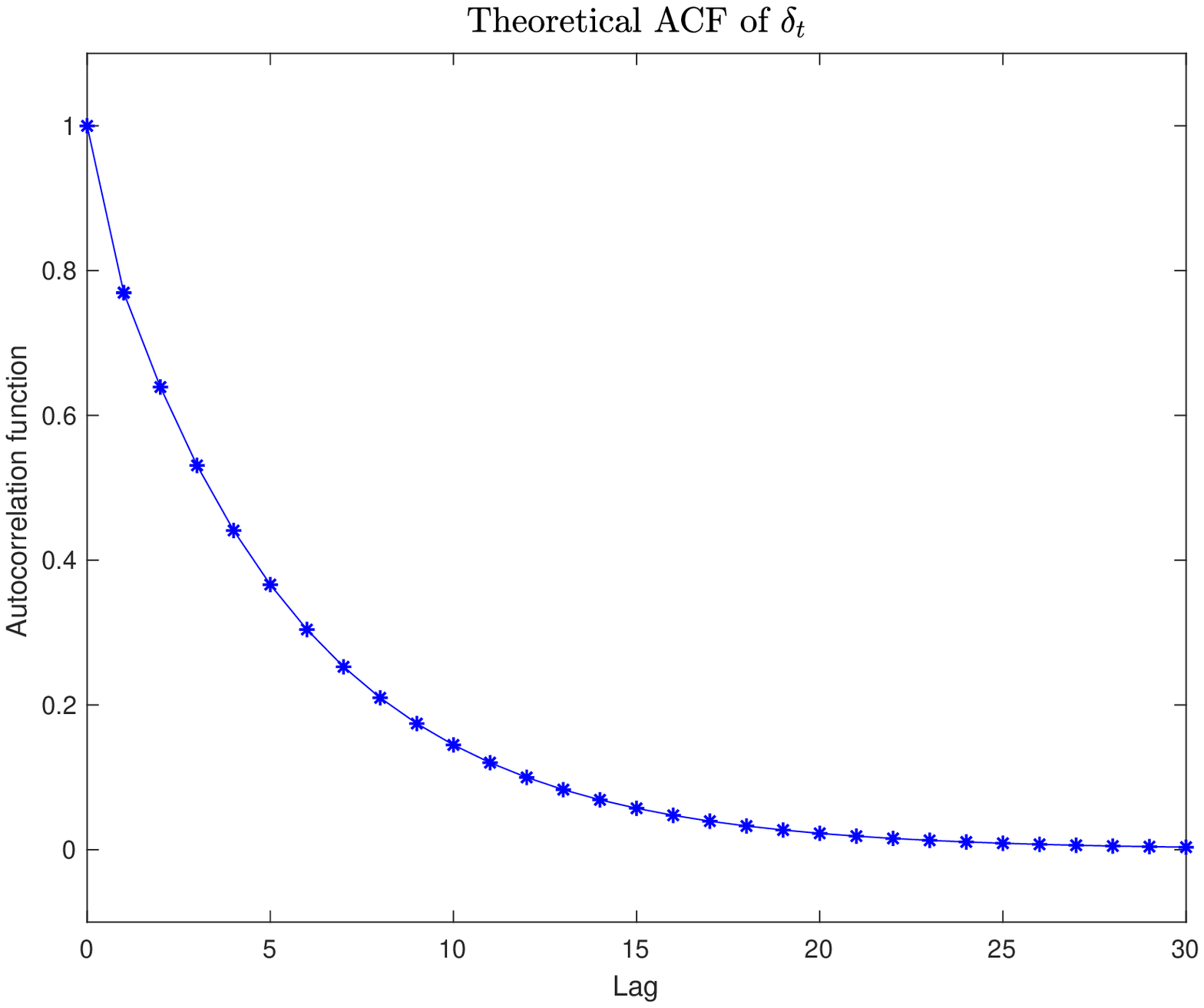}\\
\includegraphics[ height=1.800in, width=2.300in]{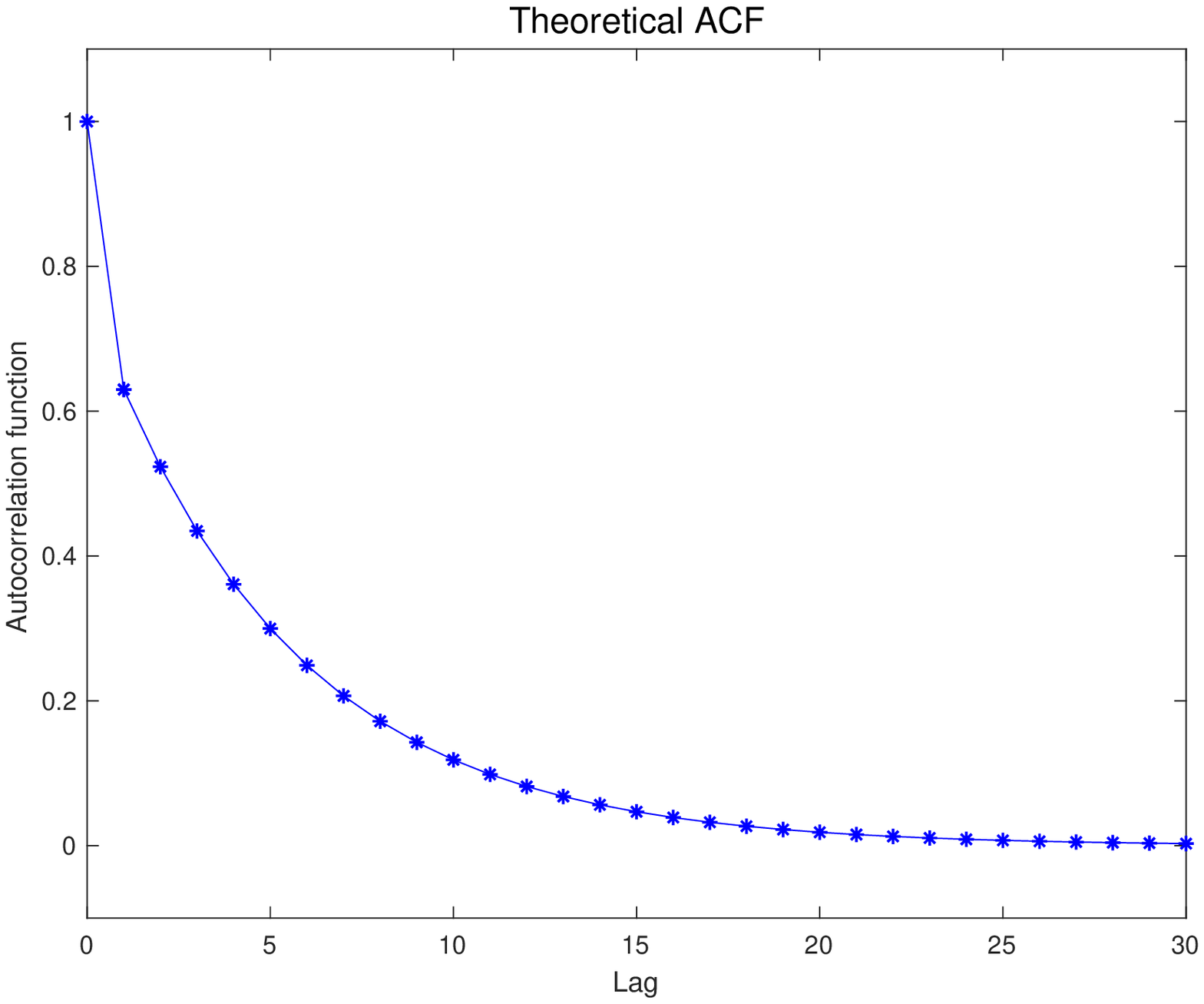}
\caption{Theoretical auto-correlation functions of Model I.}
\label{fig:realACF_1}
\end{figure}


\section{Parameter Estimation}\label{sec:mle}
In this section, we develop parameter estimation of our Int-GARCH model and the asymptotic properties of the estimators. There are two groups of parameters: the error distribution parameter $k>0$ and the variance parameters 
$$\bdsm{\theta}=\left[\mu, \alpha_1, \cdots, \alpha_p, \beta_1, \cdots, \beta_q, \gamma_1, \cdots, \gamma_w\right]^T.$$
We first provide a sample estimate for $k$, which is shown to be strongly consistent. Then the variance parameters $\bdsm{\theta}$ are estimated by the method of maximum likelihood, and the asymptotic normality for the MLE is established. A simulation study is presented that shows the empirical performance of the estimators.

\subsection{Two-stage estimation of $k$ and $\bdsm{\theta}$}
The error distribution parameter $k$ can be conveniently estimated by the method of moments.  Notice that
\begin{eqnarray*}
  E\left(\delta_t\right) &=& E\left(h_t\eta_t\right)=E\left(h_t\right)E\left(\eta_t\right)=kE\left(h_t\right),\label{ini_1}\\
  E\left|\lambda_{t}\right| &=& E\left|h_{t}\varepsilon_{t}\right|=E\left(h_{t}\right)E\left|\varepsilon_{t}\right|=\sqrt{2/{\pi}}E\left(h_{t}\right),\label{ini_2}
\end{eqnarray*}
and consequently,
\begin{equation*}
  k=\sqrt{2/\pi}\frac{E\left(\delta_t\right)}{E\left(|\lambda_t|\right)}.
\end{equation*}
Replacing $E\left(\delta_t\right)$ and $E\left(|\lambda_t|\right)$ in (\ref{mm-k}) by their sample estimates, we obtain the method of moments estimator of $k$ as
\begin{equation}\label{mm-k}
  \hat{k}=\sqrt{2/\pi}\frac{\overline{\delta_t}}{\overline{|\lambda_t|}}.
\end{equation}
Under the condition that $\left\{h_t\right\}$ is strictly stationary and ergodic, and additionally that $E\left(h_t\right)<\infty$, the sample means $\overline{\delta_t}$ and $\overline{|\lambda_t|}$ are both strongly consistent, and so is $\hat{k}$.

Given $k$, the conditional likelihood function of $\left\{r_t, t=1,\cdots, T\right\}$ is 
\begin{eqnarray*}
  L\left(\bdsm{\theta}\right)|\mathcal{F}_0
  &=&\prod\limits_{t=1}^{T}f\left(\lambda_t, \delta_t|\mathcal{F}_{t-1}\right)
  =\prod\limits_{t=1}^{T}f\left(\lambda_t|\mathcal{F}_{t-1}\right)f\left(\delta_t|\mathcal{F}_{t-1}\right)\\
  &=&\prod\limits_{t=1}^{T}\frac{1}{h_t\sqrt{2\pi}}e^{-\frac{\lambda_t^2}{2h_t^2}}\frac{1}{\Gamma(k)h_t^k}\delta_t^{k-1}e^{-\frac{\delta_t}{h_t}}\\
  &=&\prod\limits_{t=1}^{T}\frac{\delta_t^{k-1}}{\sqrt{2\pi}\Gamma(k)}h_t^{-(k+1)}e^{-\frac{\lambda_t^2}{2h_t^2}-\frac{\delta_t}{h_t}}\\
  &\propto&\prod\limits_{t=1}^{T}h_t^{-(k+1)}e^{-\frac{\lambda_t^2}{2h_t^2}-\frac{\delta_t}{h_t}}.
\end{eqnarray*} 
Thus, the conditional log-likelihood function up to a constant is 
\begin{eqnarray*}
  l(\bdsm{\theta})=\sum\limits_{t=1}^{T}\left\{-(k+1)\log(h_t)-\frac{\lambda_t^2}{2h_t^2}-\frac{\delta_t}{h_t}\right\},
\end{eqnarray*}
where $h_t=\mu+\sum\alpha_i|\lambda_{t-i}|+\sum\beta_i\delta_{t-i}+\sum\gamma_ih_{t-i}$. Then, the maximum likelihood estimate $\hat{\bdsm{\theta}}_T$ is defined as
\begin{equation}\label{def:mle}
  \hat{\bdsm{\theta}}_T=\arg\max_{\bdsm{\theta}}\left\{l(\bdsm{\theta})\right\},
\end{equation}
and it can be computed by the scoring algorithm
\begin{equation*}
  \bdsm{\theta}^{(m+1)}=\bdsm{\theta}^{(m)}-\left[\nabla^2l(\bdsm{\theta}^{(m)})\right]^{-1}\nabla l(\bdsm{\theta}^{(m)}).
\end{equation*}
Under the condition of strict stationarity and ergodicity, and some moment requirement for $\left\{h_t\right\}$, $\hat{\bdsm{\theta}}_T$ is consistent and asymptotically normal. We state the result in the following theorem. Details of the proof are deferred to the Appendix. 
\begin{theorem}\label{thm:asymp-mle}
Assume the process $\left\{h_t\right\}$ is strictly stationary and ergodic, and $E\left(h_t^2\right)<\infty$. Assume in addition that the parameter space $\Theta$ is compact. Then, the maximum likelihood estimator defined in (\ref{def:mle}) satisfies:\\
(i) $\hat{\bdsm{\theta}}_T\stackrel{\mathcal{P}}{\to}\bdsm{\theta}$ as $T\to\infty$; \\
(ii) $T^{\frac{1}{2}}\left(\hat{\bdsm{\theta}}_T-\bdsm{\theta}_0\right)\stackrel{\mathscr{D}}{\to}N\left(\bdsm{0}, \bdsm{I}^{-1}(\bdsm{\theta}_0)\right)$,
where $\bdsm{I}(\bdsm{\theta}_0)=-E\left[\nabla^2l_t(\bdsm{\theta}_0)\right]$ is the Fisher information matrix evaluated at $\bdsm{\theta}_0$.
\end{theorem} 
As an immediate consequence of the theorem, the asymptotic covariance matrix of $\hat{\bdsm{\theta}}_T$ is $\frac{1}{T}\bdsm{I}^{-1}(\bdsm{\theta}_0)$. It can be consistently estimated by the inverse Hessian $-\left[\nabla^2l\left(\hat{\bdsm{\theta}}_T\right)\right]^{-1}$, which is easily obtained from the scoring algorithm. Especially for Int-GARCH(1,1,1), the condition for strict stationarity and ergodicity is the same as that for mean stationarity, that is, $E(x_t)<\infty$, and consequently $E(h_t)<\infty$ (see Lemma \ref{stationarity-ht} in the Appendix). So, assuming mean stationarity, the method of moments estimator $\hat{k}$ is strongly consistent. If we further assume covariance stationarity, i.e., $E(x_t^2)<\infty$ in view of Corollary \ref{cor:stat}, the maximum likelihood estimator $\hat{\bdsm{\theta}}$ is consistent and asymptotically normal. We summarize the conclusion in the following corollary. 

\begin{corollary}
Consider the Int-GARCH(1,1,1) model. \\
(i) If $E(x_t)<\infty$, then $\hat{k}$ is strongly consistent; \\
(ii) If in addition $E(x_t^2)<\infty$, then $\hat{\bdsm{\theta}}$ satisfies both consistency and asymptotic normality. 
\end{corollary}

Computation of the maximum likelihood estimate requires initial values of $\bdsm{\theta}$ and starting values $\left\{h_0, \cdots, h_{-(m-1)}\right\}$ and $\left\{r_0, \cdots, r_{-(m-1)}\right\}$. Recall from Theorem \ref{thm:mean-gen} that
\begin{equation}\label{Eh_t-general}
   E\left(h_{t}\right)=\frac{\mu}{1-\sum_{i=1}^{m}\mu_i}
   =\frac{\mu}{1-\sqrt{\pi/2}\sum_{i=1}^{p}\alpha_i-k\sum_{i=1}^{q}\beta_i-\sum_{i=1}^{w}\gamma_i}.
\end{equation}
So $\mu$ can be initialized by replacing $E\left(h_{t}\right)$ in (\ref{Eh_t-general}) by its sample mean $\overline{h_t}$ and a rough guessing of
$1-\sqrt{\pi/2}\sum_{i=1}^{p}\alpha_i-k\sum_{i=1}^{q}\beta_i-\sum_{i=1}^{w}\gamma_i$, for example, 0.4. Additionally, initial values of $\alpha_i$'s, $\beta_i$'s, and $\gamma_i$'s can be obtained by setting each of $\sqrt{\pi/2}\sum_{i=1}^{p}\alpha_i$, $k\sum_{i=1}^{q}\beta_i$, and $\sum_{i=1}^{w}\gamma_i$ to be a small value,  e.g., 0.2, and by letting $\alpha_1=\cdots=\alpha_p$, $\beta_1=\cdots=\beta_q$, $\gamma_1=\cdots=\gamma_w$. 
Finally, we assume that the process $\left\{r_t\right\}$ starts from its infinite past with a finite mean and variance, so it is reasonable to let $h_t=E\left(h_t\right)$, $t=0, \cdots, -(m-1)$. An alternative is to let $h_t=0$, $t=0, \cdots, -(m-1)$, assuming $\left\{r_t\right\}$ starts from constant intervals $\left\{r_0, \cdots, r_{-(m-1)}\right\}$. In either case, we let $r_t=E\left(r_t\right)$, $t=0, \cdots, -(m-1)$. Based on our experience, the initial values lead to very fast convergence, with only a couple of iterations. 

Alternative to the two-stage estimation, one can also compute the maximum likelihood estimates of $k$ and $\bdsm{\theta}$ simultaneously, which jointly maximize the conditional likelihood
\begin{equation*}
   L\left(k, \bdsm{\theta}\right)|\mathcal{F}_0
   \propto \prod\limits_{t=1}^{T}  \frac{1}{\Gamma(k)}\delta_t^{k}h_t^{-(k+1)}e^{-\frac{\lambda_t^2}{2h_t^2}-\frac{\delta_t}{h_t}},
\end{equation*}
or equivalently the conditional log-likelihood up to a constant
\begin{equation*}
    l(k,\bdsm{\theta})=\sum\limits_{t=1}^{T}
    \left\{k\log(\delta_t)-\log\left[\Gamma(k)\right]-(k+1)\log(h_t)-\frac{\lambda_t^2}{2h_t^2}-\frac{\delta_t}{h_t}\right\}.
\end{equation*}
For the iterative algorithm, the two-stage estimates can be naturally input as initial values. Due to the involvement of the gamma function in the conditional likelihood, computation of this joint maximum likelihood estimates will be more complicated. Additionally, the required conditions to establish the asymptotic normality are also expected to be much more restrictive.  

\subsection{Simulation and finite sample performances}\label{simulation}
Our simulation study considers four distinct Int-GARCH models with specific sets of parameter values.
From the empirical analysis detailed in Section 6, there seem to be certain ranges for the parameters: 1) $k$ is between $1$ and $2$; 2) $\mu$ and $\alpha_1$ are very small; 
3) $\beta_1$ is generally above $0.3$; 4) $\gamma_1$ is either very small or equal to $0$. We design our models such that the parameter values are inside these ranges,  
with a little added flexibility to account for potential possibilities not included in our real data. We are particularly interested in the Int-ARCH model, since very often in the real 
data analysis $\gamma_1$ is estimated to be $0$. So, after examining two full Int-GARCH models (Models I and II), we also consider two Int-ARCH models with $\gamma_1$ 
set to $0$ (Models III and IV). The parameter values for the two Int-GARCH models are generated as follows. 
\begin{align*}
\text{Models I and II}:\ \ \ \ 
  &k\sim\text{Unif }(1, 2),\\
  &\mu\sim\text{Unif }(0, 0.1),\\
  &\alpha_1\sim\text{Unif }(0, 0.2),\\
  &\beta_1\sim\text{Unif }(0.1, 0.6),\\
  &\gamma_1\sim\text{Unif }(0, 0.2). 
\end{align*} 
For Int-ARCH models, $\beta_1$ is generally larger due to the removal of $\gamma_1$, so we generate the parameters for the two Int-ARCH models with $k, \mu, \alpha_1$ being the same as the Int-GARCH models but $\beta_1$ slightly larger. 
\begin{align*}
\text{Models I and II}:\ \ \ \ 
  &k\sim\text{Unif }(1, 2),\\
  &\mu\sim\text{Unif }(0, 0.1),\\
  &\alpha_1\sim\text{Unif }(0, 0.2),\\
  &\beta_1\sim\text{Unif }(0.3, 0.6).
\end{align*} 
All of the generations are subject to the constraint that each combination will result in a weakly stationary Int-(G)ARCH process that achieves consistency and asymptotic normality for the maximum likelihood estimators. The exact parameter values are listed in Table \ref{tab:simu}. Realizations of the designed models are simulated using the initial values $h_{0}=0$ and $r_{0}=E\left(r_{t}\right)$.  Plots of two simulated data sets are shown in Figure \ref{fig:data_simu}. Denote
\begin{table}[htbp]
  \centering
  \caption{Average result of 100 repeated simulations from two Int-GARCH (1,1,1) models (I, II) and two Int-ARCH (1,1) models (III, IV). }
    \begin{tabular}{rllllll}
    \toprule
          &       &       & \multicolumn{1}{c}{Mean} & \multicolumn{1}{c}{Mean} & \multicolumn{1}{c}{Empiriccal} & \multicolumn{1}{c}{Asymptotic} \\
    \multicolumn{1}{c}{Model} & \multicolumn{2}{c}{Parameters} & \multicolumn{1}{c}{Estimate} & \multicolumn{1}{c}{Absolute Error} & \multicolumn{1}{c}{Standard Error } & \multicolumn{1}{c}{Standard Error} \\
    \midrule
          &       &       &       &       &       &  \\
    \multicolumn{1}{c}{\textbf{I}} & $k$     & 1.8147 & 1.8081 & 0.077 & 0.1061 &  \\
          & $\mu$    & 0.0906 & 0.0887 & 0.0072 & 0.0086 & 0.0082 \\
          & $\alpha_1$ & 0.0318 & 0.0284 & 0.0184 & 0.0235 & 0.0211 \\
          & $\beta_1$  & 0.374 & 0.3586 & 0.0171 & 0.0148 & 0.0149 \\
          & $\gamma_1$ & 0.1265 & 0.1269 & 0.0314 & 0.0381 & 0.0314 \\
          &       &       &       &       &       &  \\
    \multicolumn{1}{c}{\textbf{II}} & $k$     & 1.2134 & 1.2251 & 0.0412 & 0.05  &  \\
          & $\mu$    & 0.071 & 0.0712 & 0.0068 & 0.0086 & 0.0094 \\
          & $\alpha_1$ & 0.1833 & 0.1784 & 0.025 & 0.0318 & 0.0306 \\
          & $\beta_1$  & 0.2334 & 0.2345 & 0.0152 & 0.0192 & 0.0207 \\
          & $\gamma_1$ & 0.1732 & 0.174 & 0.0467 & 0.0579 & 0.0602 \\
          &       &       &       &       &       &  \\
    \multicolumn{1}{c}{\textbf{III}} & $k$     & 1.5139 & 1.5045 & 0.04  & 0.0506 &  \\
          & $\mu$    & 0.074 & 0.0738 & 0.0026 & 0.0034 & 0.0036 \\
          & $\alpha_1$ & 0.037 & 0.0334 & 0.0185 & 0.0228 & 0.0208 \\
          & $\beta_1$  & 0.3436 & 0.3426 & 0.0139 & 0.017 & 0.0174 \\
          & $\gamma_1$ & 0     &       &       &       &  \\
          &       &       &       &       &       &  \\
    \multicolumn{1}{c}{\textbf{IV}} & $k$     & 1.3632 & 1.3621 & 0.038 & 0.0486 &  \\
          & $\mu$    & 0.0584 & 0.0593 & 0.0029 & 0.0037 & 0.0033 \\
          & $\alpha_1$ & 0.1927 & 0.1962 & 0.0208 & 0.026 & 0.0269 \\
          & $\beta_1$  & 0.322 & 0.3253 & 0.0161 & 0.0197 & 0.0195 \\
          & $\gamma_1$ & 0     &       &       &       &  \\
    \bottomrule
    \end{tabular}%
  \label{tab:simu}%
\end{table}

\begin{figure}[ht]
\centering
\includegraphics[ height=1.800in, width=2.300in]{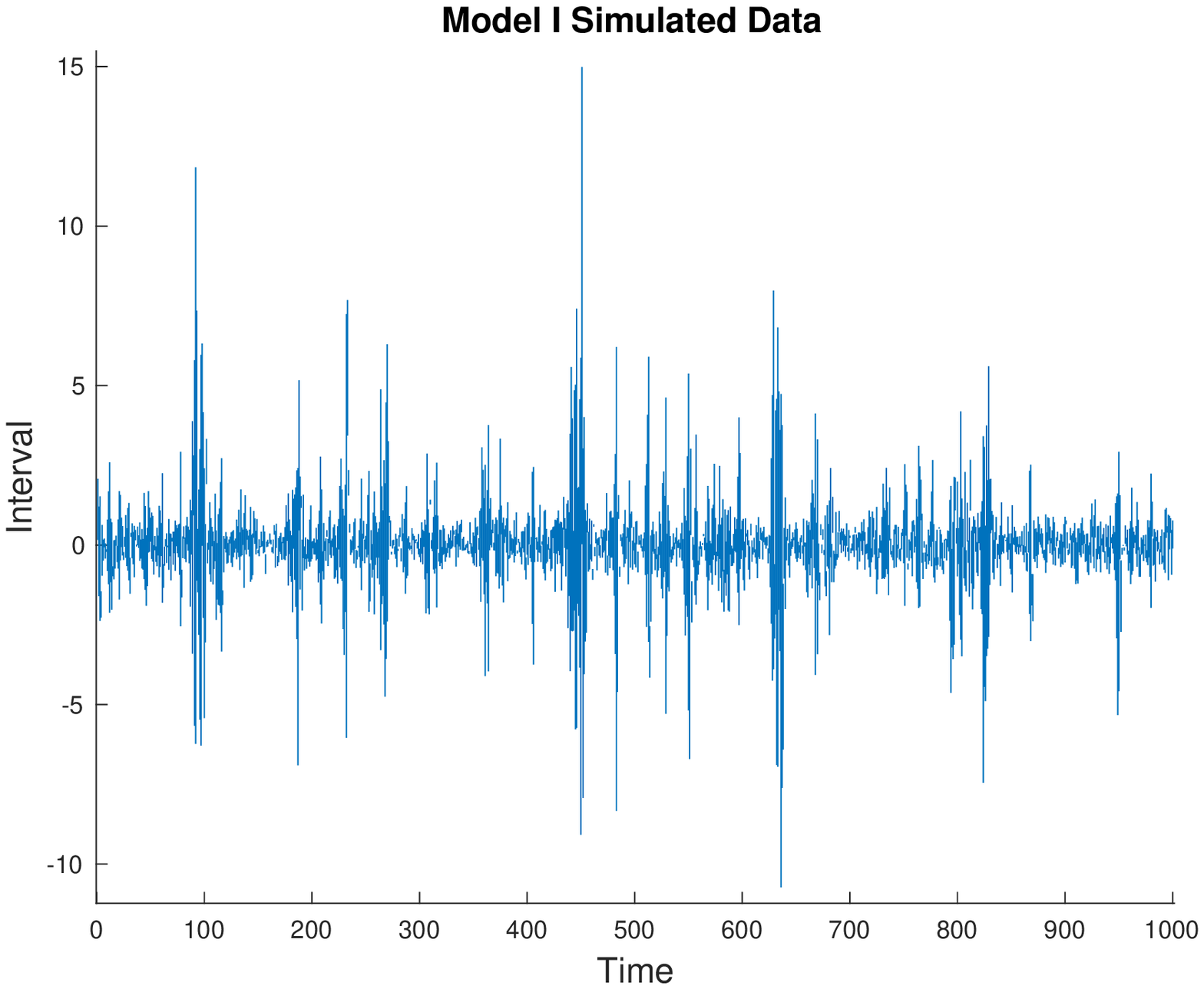}
\includegraphics[ height=1.800in, width=2.300in]{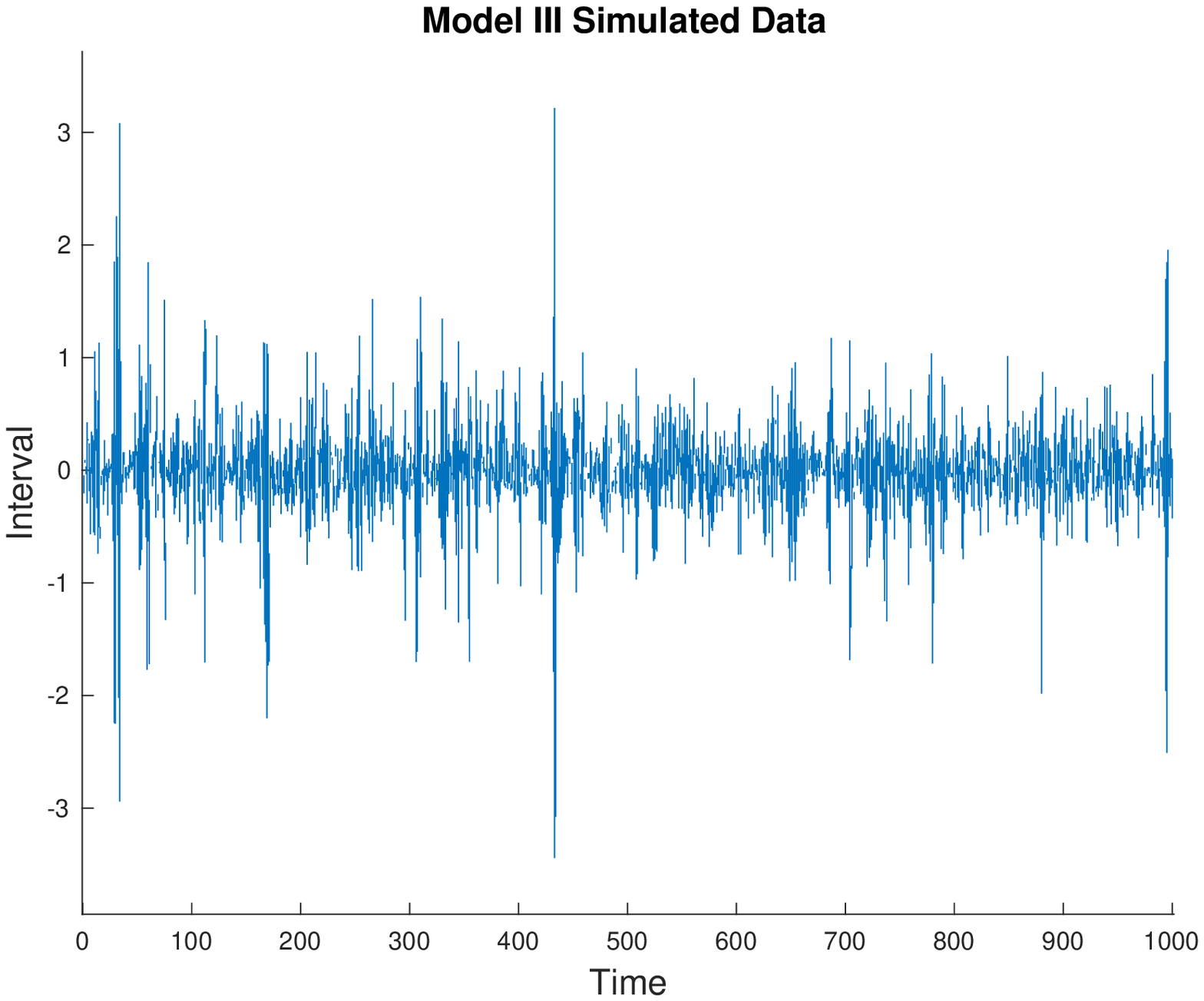}
\caption{Plots of simulated data sets each with $T=1000$.}
\label{fig:data_simu}
\end{figure}

\begin{eqnarray*}
  \gamma(s) &=& \mbox{Cov}\left(r_{t},r_{t+s}\right),\\
  \gamma_{\lambda}(s) &=& \mbox{Cov}\left(\lambda_{t},\lambda_{t+s}\right),\\
  \gamma_{\delta}(s) &=& \mbox{Cov}\left(\delta_{t},\delta_{t+s}\right).
\end{eqnarray*}
Recall that the theoretical ACF of $\left\{r_t\right\}$ is
\begin{eqnarray*}
\rho(s) & = & \dfrac{\gamma(s)}{\gamma(0)}=\dfrac{\gamma_{\lambda}(s)+\gamma_{\delta}(s)}{\gamma_{\lambda}(0)+\gamma_{\delta}(0)}.
\end{eqnarray*}
We consequently define the sample ACF of $\left\{ r_{t}\right\} $ as
\[
\rho(s)=\dfrac{\hat{\gamma_{\lambda}}(s)+\hat{\gamma_{\delta}}(s)}{\hat{\gamma_{\lambda}}(0)+\hat{\gamma_{\delta}}(0)},
\]
where $\hat{\gamma_{\lambda}}(s)$ and $\hat{\gamma_{\delta}}(s)$ are the sample auto-covariance functions of $\left\{ \lambda_{t}\right\}$ and $\left\{ \delta_{t}\right\} $, respectively. Figures \ref{fig:acf_1} and \ref{fig:acf_2} show the sample ACF's for simulated data sets with 1000 observations from an Int-GARCH(1,1,1) (Model I) and an Int-ARCH(1,1) (Model III), respectively.\\

\begin{figure}[ht]
\centering
\includegraphics[ height=1.800in, width=2.300in]{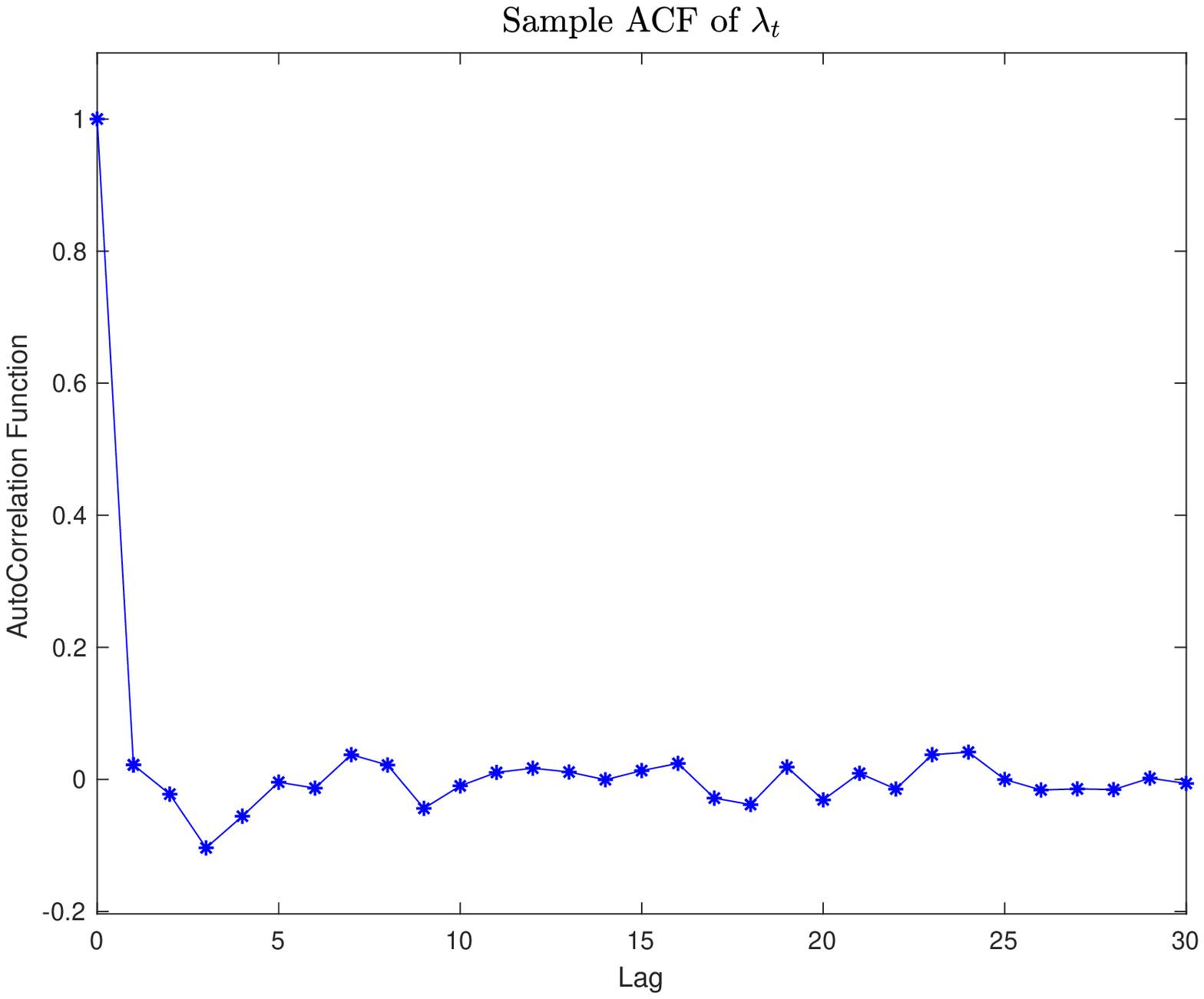}
\includegraphics[ height=1.800in, width=2.300in]{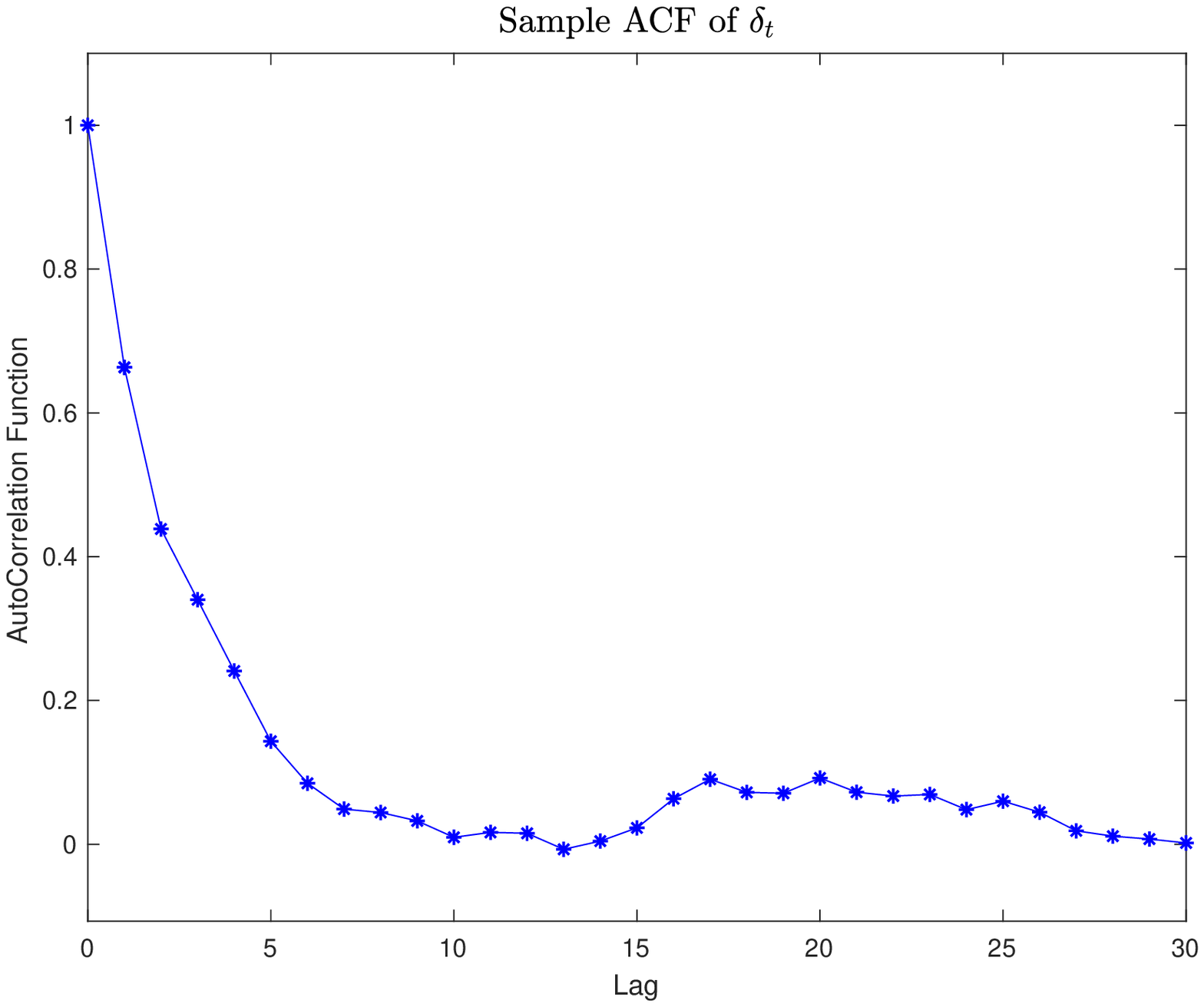}\\
\includegraphics[ height=1.800in, width=2.300in]{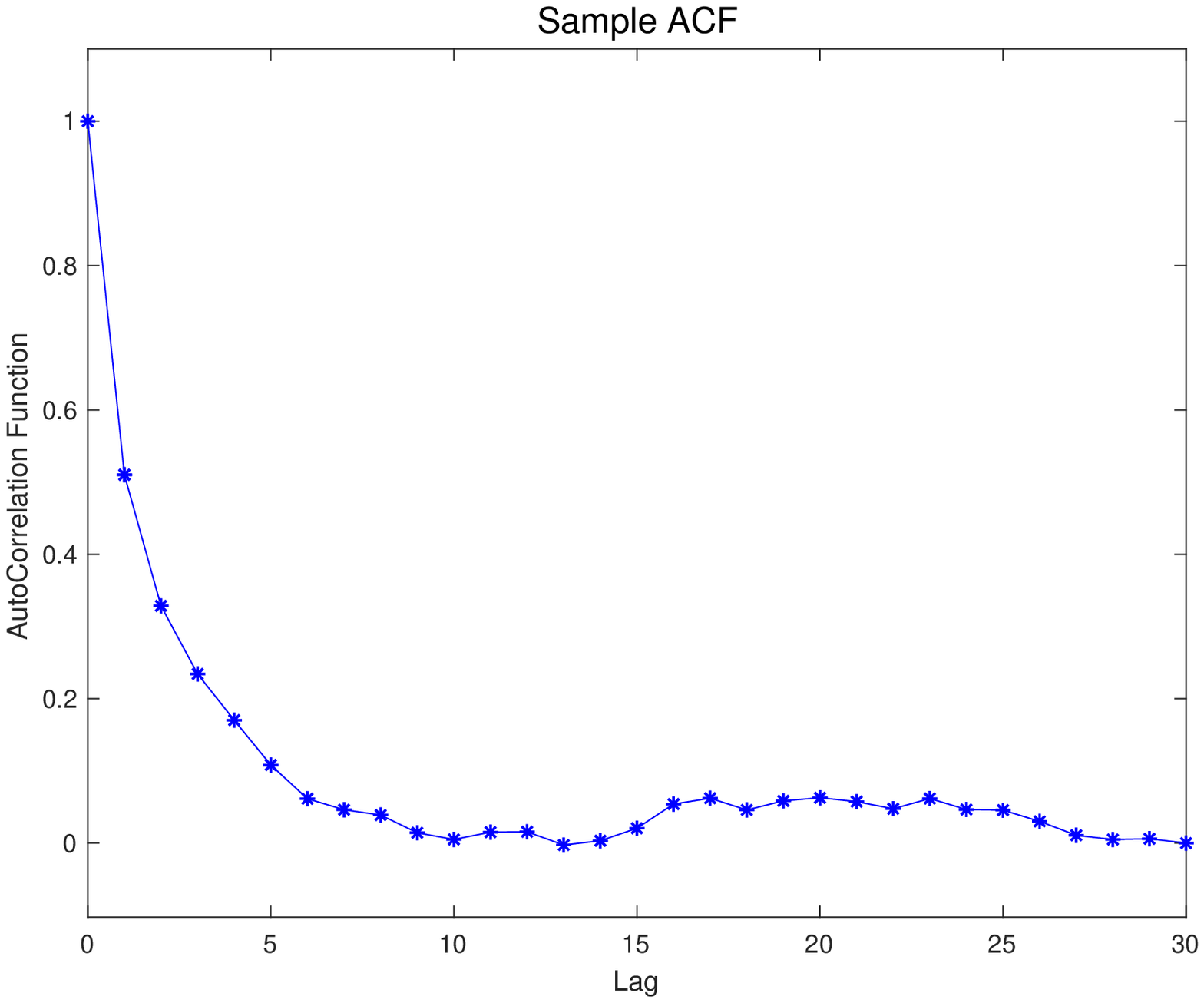}
\caption{Sample auto-correlation functions of a simulated data set from Model I.}
\label{fig:acf_1}
\end{figure}

\begin{figure}[ht]
\centering
\includegraphics[ height=1.800in, width=2.300in]{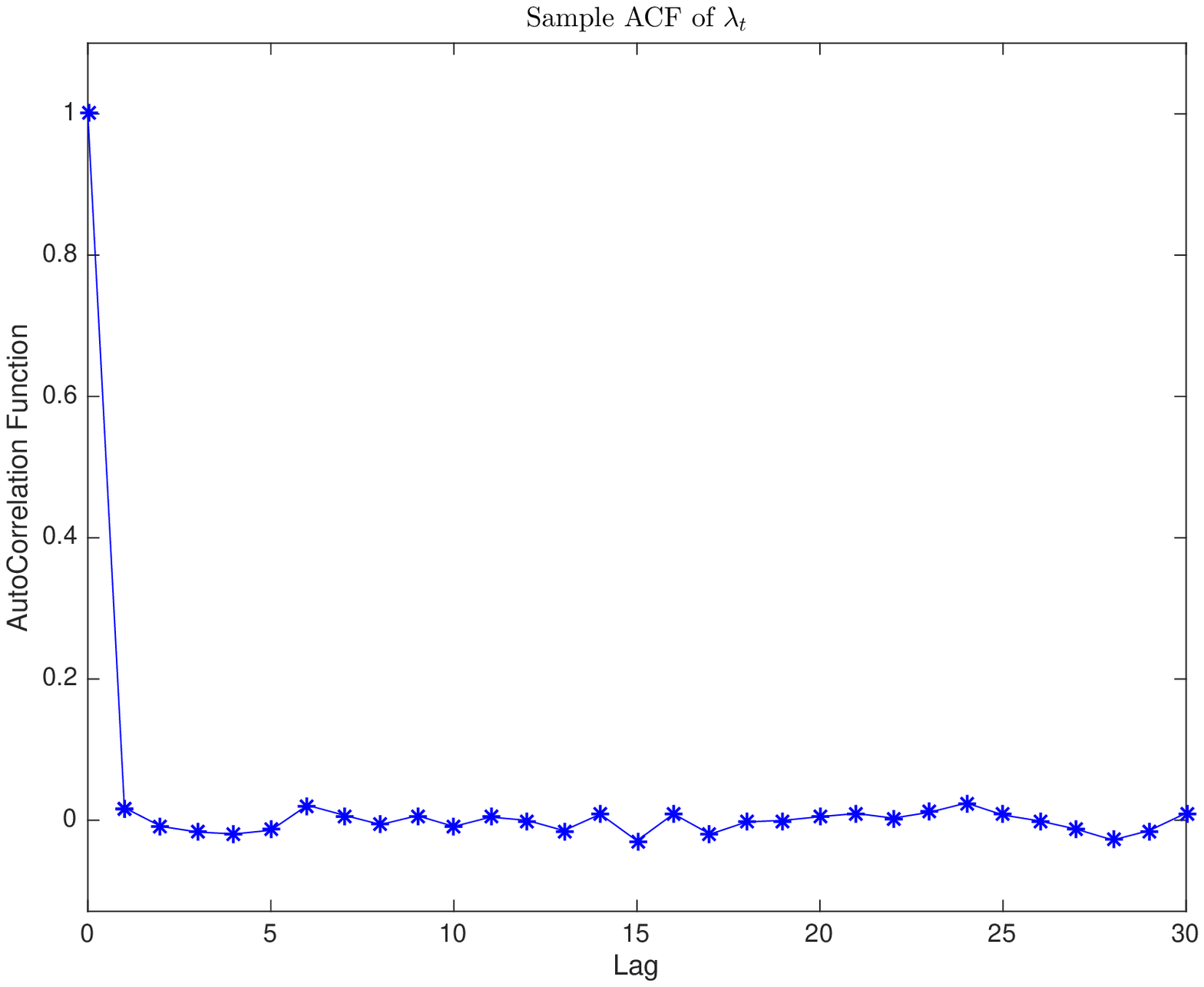}
\includegraphics[ height=1.800in, width=2.300in]{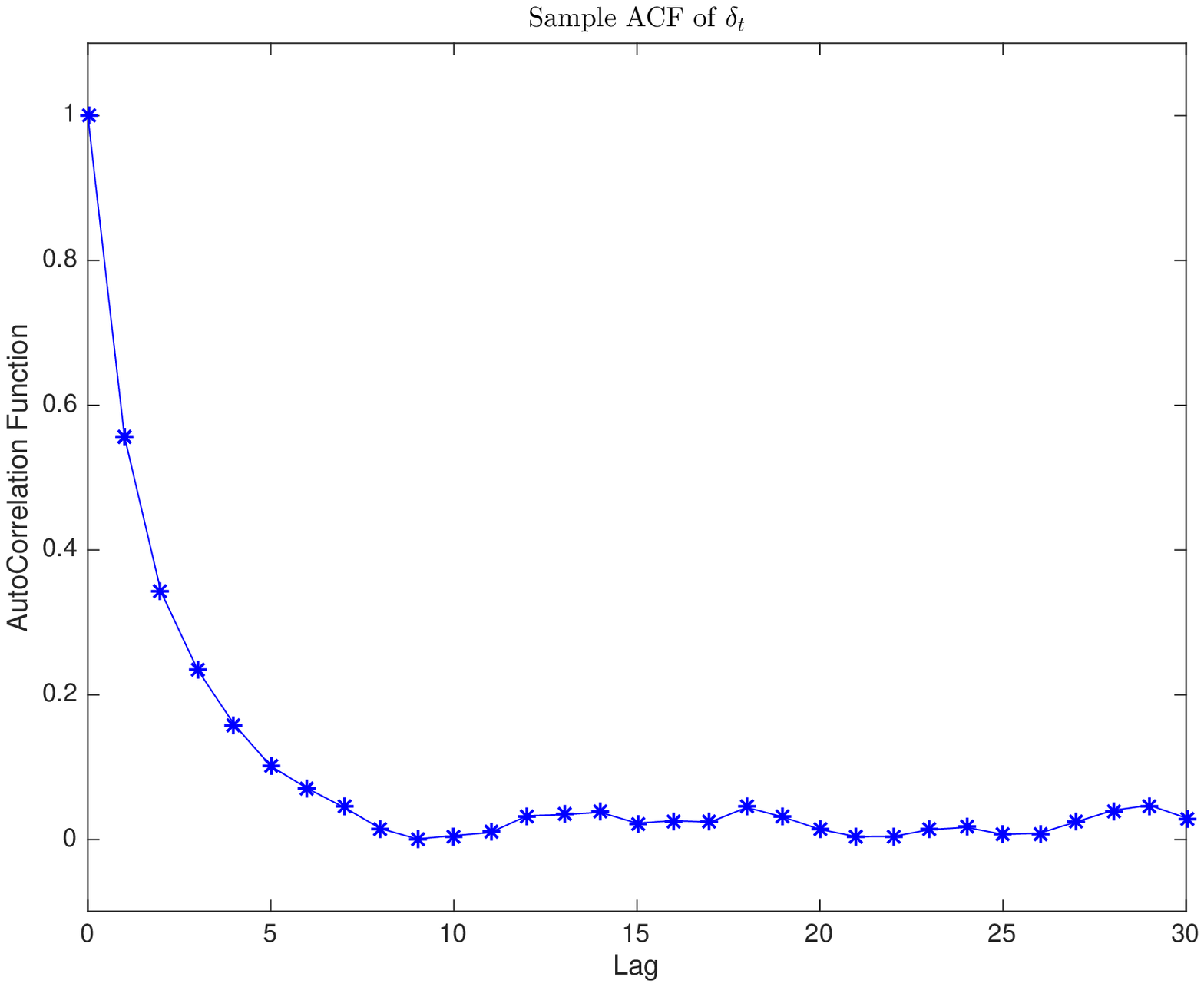}\\
\includegraphics[ height=1.800in, width=2.300in]{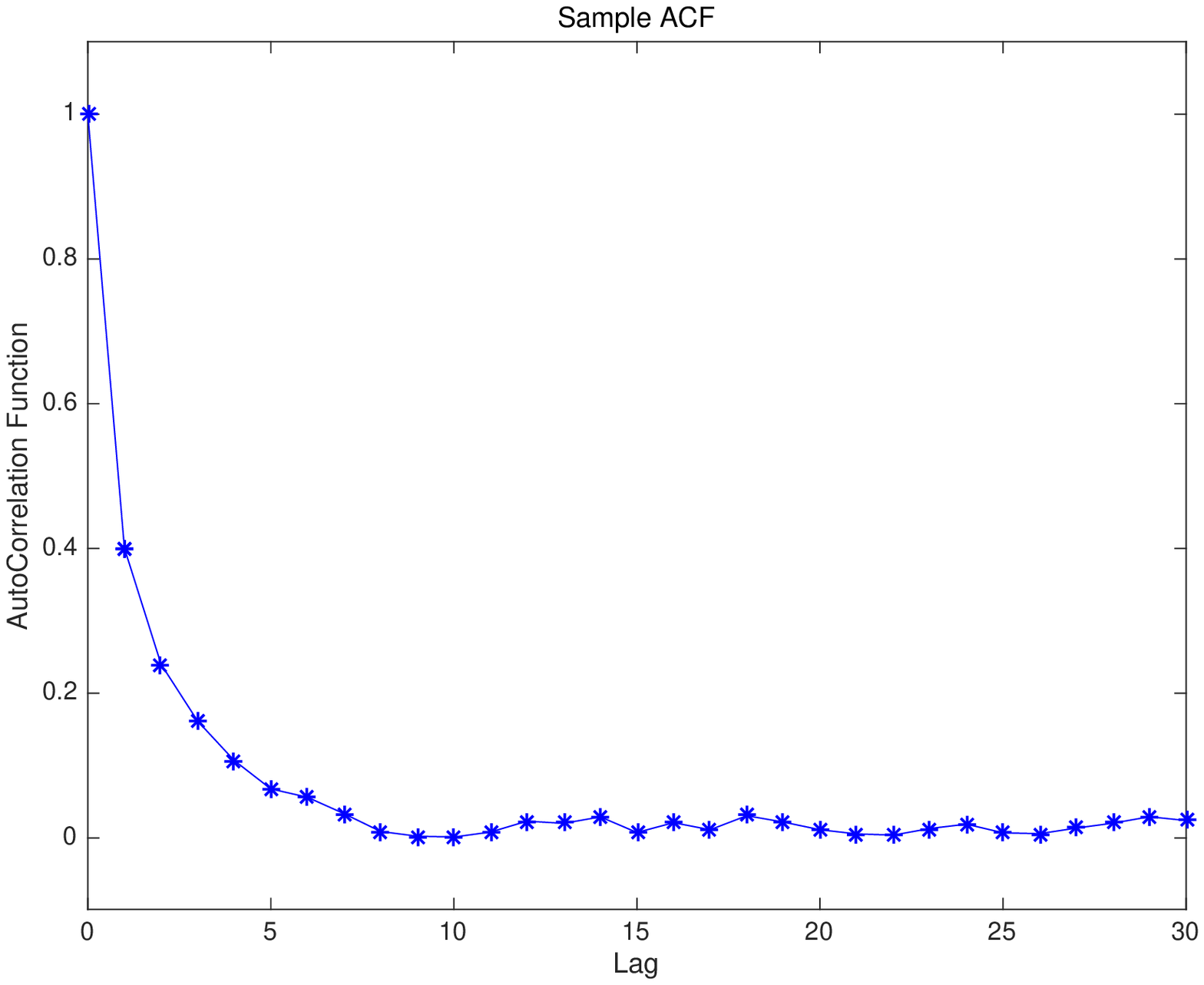}
\caption{Sample auto-correlation functions of a simulated data set from Model III.}
\label{fig:acf_2}
\end{figure}

\indent For each of the four models, we simulate a data set with 1000 observations and estimate the parameters using the proposed two-stage procedure. The process of simulation and estimation is repeated for 100 times independently, and the average results are reported in Table \ref{tab:simu}. We see that, with a very moderate sample of size $1000$, both the moment estimate for $k$ and the maximum likelihood estimate for $\bdsm{\theta}$ are very close to the true values with small mean absolute errors. In addition, the asymptotic standard errors for $\bdsm{\theta}$ based on the asymptotic normality are compared to the empirical ones, and they match very well. 



\section{Empirical Application }\label{sec:empirical}
In this section, we apply our Int-GARCH model to analyze 10 typical stocks and 4 indices data. The symbols of these stocks and indices are listed, for example, in Table \ref{tab:est_model}. We obtained from Thomson Reuters both the 5-minute intraday data and the daily-closing data of totally 1511 trading days, from January 3, 2006 to December 30, 2011. As an example for demonstration purpose, the sample ACF of the BA (Bank of America) return range is displayed in Figure \ref{fig:acf-BA}. It is very similar to the theoretical ACF of model I in Figure \ref{fig:realACF_1} from Section \ref{simulation}, which indicates the feasibility of our Int-GARCH model.

\begin{figure}[ht]
\centering
\includegraphics[ height=1.800in, width=2.300in]{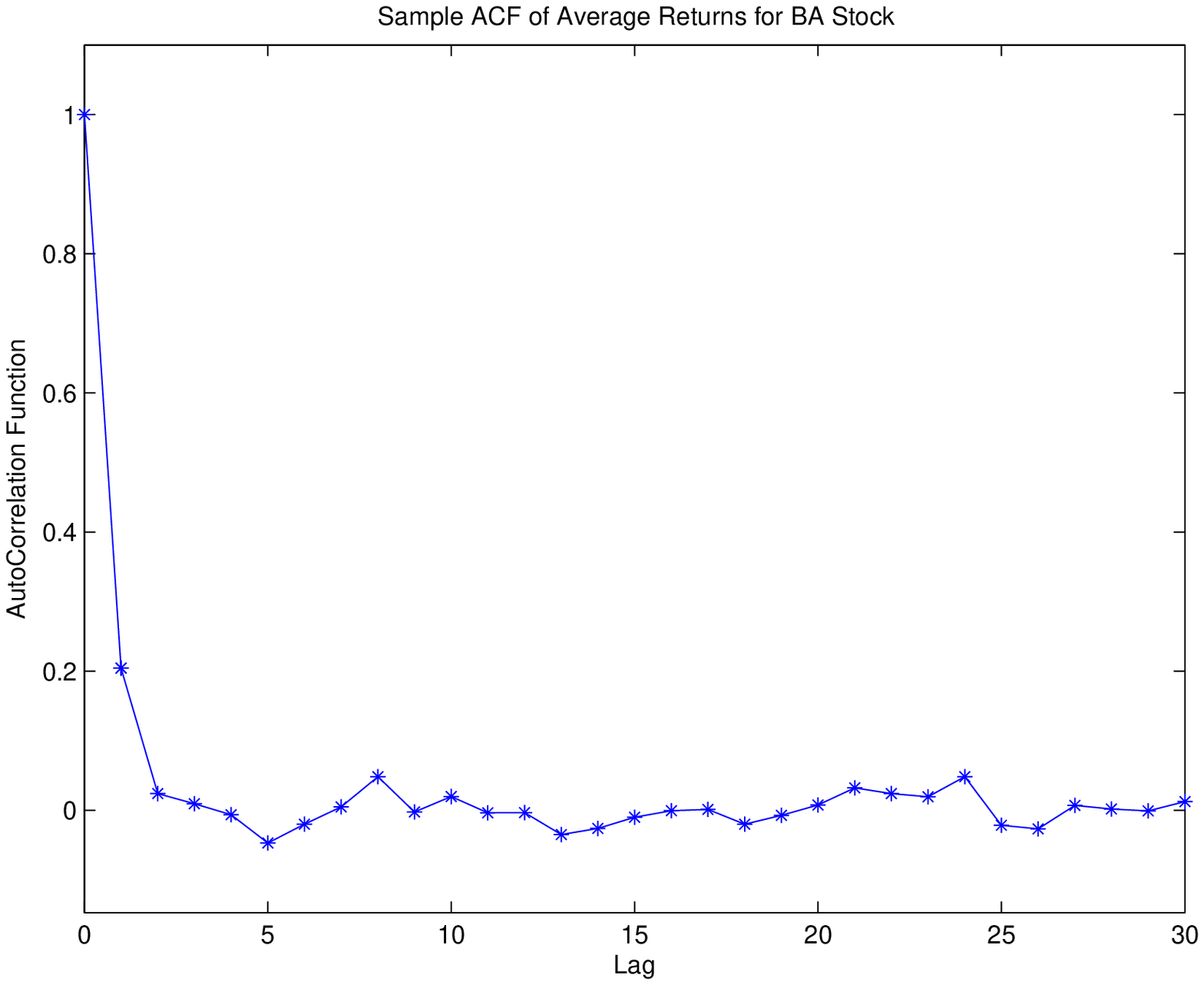}
\includegraphics[ height=1.800in, width=2.300in]{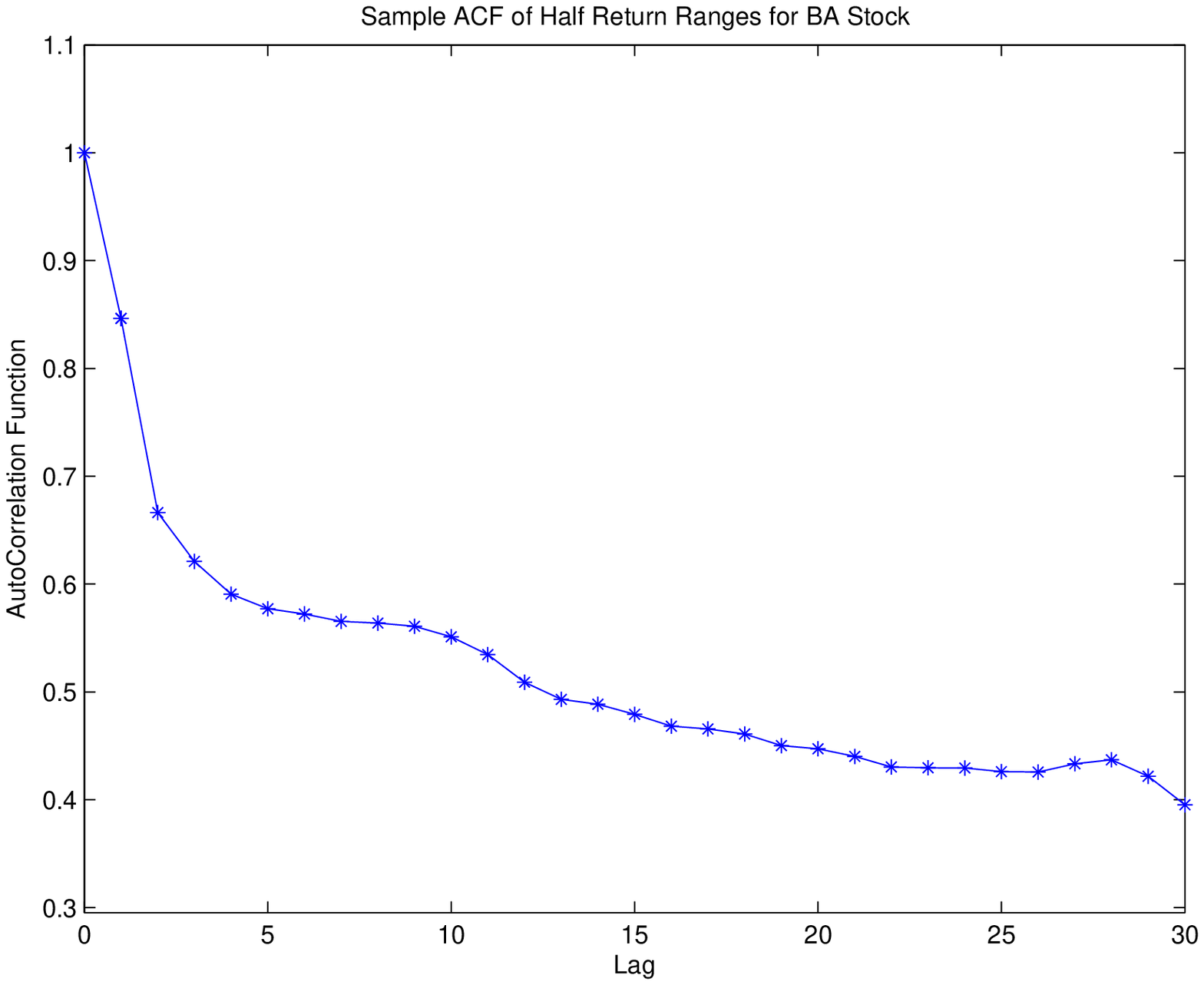}\\
\includegraphics[ height=1.800in, width=2.300in]{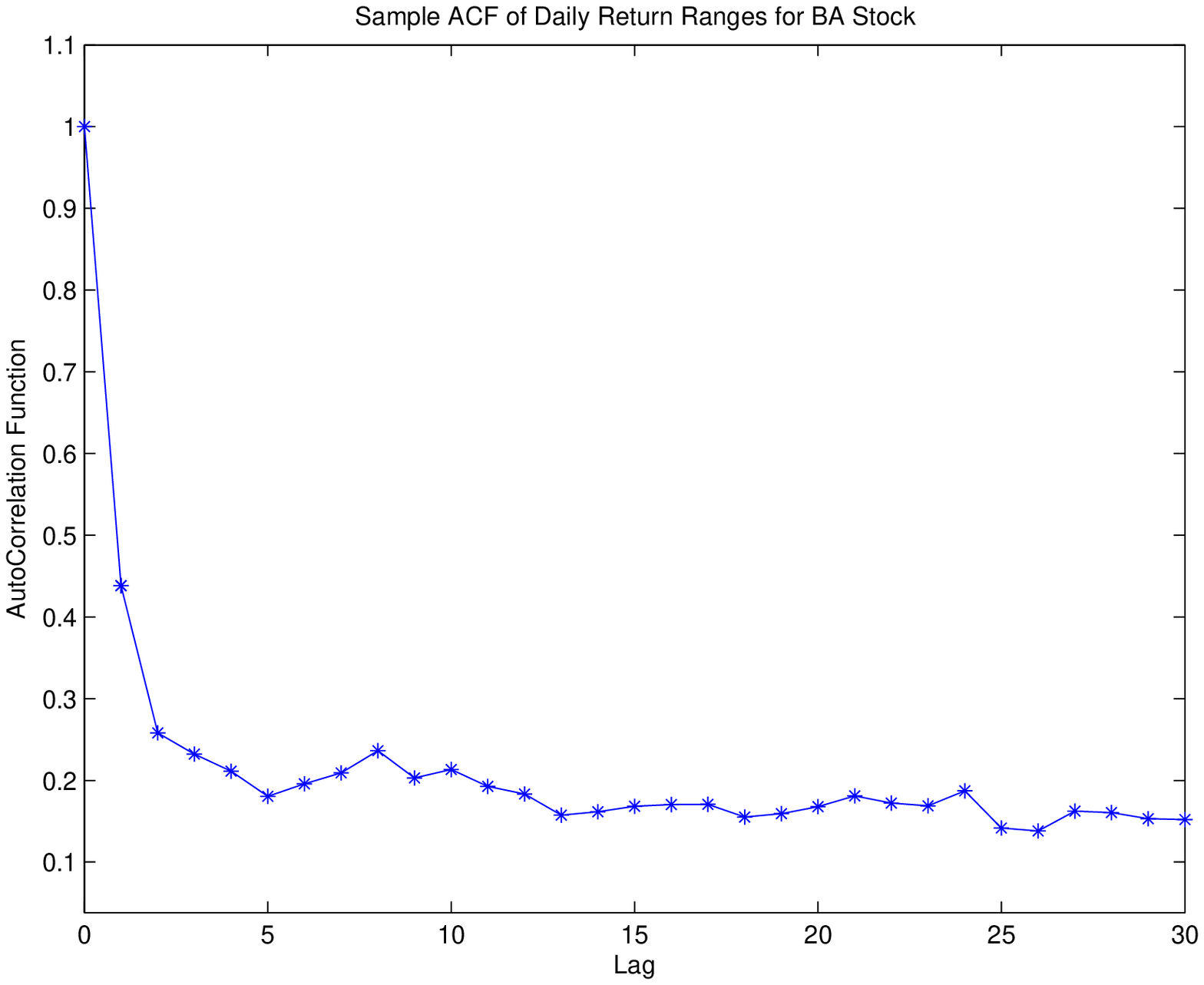}
\caption{Sample auto-correlation function of BA stock.}
\label{fig:acf-BA}
\end{figure}

\begin{table}[htbp]
  \centering
  \caption{Fitted Int-(G)ARCH models for the period 2006-2011. For the variance parameters, the number in the parenthesis to the right of the estimate is the associated standard error 
  based on the asymptotic normality.}
    \begin{tabular}{rlrlllll}
    \toprule
          &       &       & \multicolumn{5}{c}{\textbf{Parameter Estimates}} \\
          &       &       & \multicolumn{1}{c}{k} & \multicolumn{1}{c}{$\mu$} & \multicolumn{1}{c}{$\alpha_1$} & \multicolumn{1}{c}{$\beta_1$} & \multicolumn{1}{c}{$\gamma_1$} \\
    \midrule
          &       &       &       &       &       &       &  \\
    \multicolumn{1}{l}{\textbf{Stocks}} & AAPL  &       & 1.5694 & 0.0043 (0.0004) & 0.024 (0.0242) & 0.4948 (0.0233) & 0 \\
          & AXP   &       & 1.7676 & 0.0012 (0.0003) & 0     & 0.5357 (0.0147) & 0 \\
          & BA    &       & 1.6245 & 0.0031 (0.0005) & 0.0616 (0.0239) & 0.4820 (0.0258) & 0 \\
          & BAC   &       & 1.511 & 0.00003 (0.0003) & 0.0374 (0.0266) & 0.6507 (0.0214) & 0 \\
          & DD    &       & 1.7335 & 0.0019 (0.0004) & 0.0034 (0.0240) & 0.5078 (0.0226) & 0 \\
          & JPM   &       & 1.7228 & 0.0006 (0.0004) & 0.0139 (0.0203) & 0.3399 (0.0096) & 0.3894 (0.0354) \\
          & KO    &       & 1.8602 & 0.0014 (0.0002) & 0.0336 (0.0227) & 0.4480 (0.0212) & 0 \\
          & MSFT  &       & 1.7135 & 0.0029 (0.0004) & 0.0708 (0.0257) & 0.4485 (0.0238) & 0 \\
          & T     &       & 1.8338 & 0.0022 (0.0003) & 0     & 0.4510 (0.0181) & 0 \\
          & WMT   &       & 1.8777 & 0.0017 (0.0003) & 0     & 0.4540 (0.0200) & 0 \\
          &       &       &       &       &       &       &  \\
    \multicolumn{1}{l}{\textbf{Indices}} & DJI   &       & 1.7755 & 0.0009 (0.0002) & 0     & 0.5051 (0.0172) & 0 \\
          & SPX   &       & 1.6472 & 0.0008 (0.0002) & 0     & 0.5383 (0.0348) & 0.0252 (0.0586) \\
          & FTSE  &       & 1.7529 & 0.0008 (0.0002) & 0     & 0.5100 (0.0183) & 0 \\
          & CAC   &       & 1.4469 & 0.0014 (0.0003) & 0     & 0.5800 (0.0382) & 0.0265 (0.0608) \\
    \bottomrule
    \end{tabular}%
  \label{tab:est_model}%
\end{table}

\subsection{Data cleaning and ranking metrics}
In order to compare the empirical (in-sample and out-of-sample) performances of our Int-GARCH model  and the GARCH model, we use realized variance (RV) as the proxy for the ``true'' market variance. The RV using our 5-minute intraday data is computed as:
\begin{equation}\label{Definition_RV}
\textmd{ RV}_t = \sum_{s} [y_t (s) - y_t (s-1)]^2. 
\end{equation}
As before, $y_t(s)$ denotes the log price of an asset at time $s$ on day $t$, and $y_t (s-1)$ denotes the log price of the previous ``snapshot'' which was taken 5 minutes ago.

Following \cite{Shephard10}, we applied 4 filtration rules to clean the data.
1) When multiple quotes have the same timestamp, we replace all these with a single entry with
the median bid and median ask price.
2) Delete entries for which the spread is negative.
3) Delete entries for which the spread is more than 50 times the median spread on that day.
4) Delete entries for which the mid-quote deviated by more than 10 mean absolute deviations from
a rolling median centred but excluding the observation under consideration of 50 observations
(i.e. 25 observations before and 25 after).

Denote the RV by $v_t$ and estimated volatility by $\hat{\sigma}^2_t$. It is shown in \cite{Meddahi01} that the ranking of volatility forecasts based on the $R^2$ from the Mincer-Zarnowitz regression
\begin{equation*}
  v_t=\beta_0+\beta_1\hat{\sigma}_t^2+\epsilon_t,
\end{equation*}
is robust to noises in $v_t$. It is shown in \cite{Patton11} that the ranking using the negative quasi-likelihood (QLIKE) as the loss function is also robust to noises in the volatility proxy. Additionally, the heteroskedasticity-adjusted MSE (HMSE), although not robust to noises, is a popularly used loss function to compare volatility models. The definitions for the QLIKE and HMSE loss functions are given as:
\begin{eqnarray*}
  \text{QLIKE}&=&\frac{1}{N}\sum_{t=1}^{N}\left[\log\left(\hat{\sigma}^2_t\right)+\frac{v^2_t}{\hat{\sigma}^2_t}\right],\\
  \text{HMSE}&=&\frac{1}{N}\sum_{t=1}^{N}\left(\frac{v^2_t}{\hat{\sigma}^2_t}-1\right),
\end{eqnarray*}
where $N$ is the length of sampling period.

\subsection{In-sample and out-of-sample results}
First, we compare the in-sample volatility estimation with that of GARCH (1,1). For this purpose, we fit an Int-GARCH model to the entire data from January 3, 2006 to December 30, 2011. The parameters are estimated by the two-stage estimation method we proposed in Section \ref{sec:mle} and results are listed in Table \ref{tab:est_model}. Occasionally, the maximum likelihood estimate of the full model contains negative values. This means that the maximum of the likelihood function subject to the nonnegative constraints happens on the boundary, i.e., at least one of the estimates is $0$. This is equivalent to removing the corresponding parameter(s) and running the unconstrained estimation for the sub-model (see, e.g. \cite{Waterman74}). For such cases, the removed parameters are displayed as $0$ exactly
without a standard error. In general, the fitted models show that $\alpha$ is much smaller than $\beta$ both in magnitude and in statistical significance, indicating that, consistent to our expectation, the return range is of much more importance than a single ``snapshot'' return, in terms of their contributions to the volatility. Additionally, $\gamma$ also tends to be small, so most likely an Int-ARCH model is sufficient. The comparisons of Int-GARCH (1,1,1) and GARCH (1,1) using all the three criteria ($R^2$, QLIKE and HMSE) are shown in Table \ref{tab:in-sample}, which generally favor Int-GARCH, except for a few exceptions (i.e., $R^2$ of BA, HMSE of T).

Next, to examine the out-of-sample performance, we use the first 5 years (2006-2010) as training period and leave the last year (2011) for out-of-sample predictions. We consider 3 forecasting horizons: 1-step, 2-step, and 5-step, corresponding to one day, two day, and one week ahead predictions, respectively. The model parameters are updated daily and the h-step-ahead predictions (h=1,2,5) are calculated according to equations (\ref{def:int-vol})-(\ref{1-step-pred}). The predicts are compared to those from GARCH (1,1) model for the entire year of 2011, and the average results are shown in Tables \ref{tab:out-sample-1}, \ref{tab:out-sample-2}, and \ref{tab:out-sample-5}, for h=1,2,5, respectively. Similar to the in-sample comparison, the out-of-sample comparisons in Tables \ref{tab:out-sample-1}, \ref{tab:out-sample-2}, and \ref{tab:out-sample-5} indicate the Int-GARCH(1,1,1) perform better than GARCH (1,1) for 1-step, 2-step, and 5-step ahead predictions, for most of the stocks (indices) and ranking metrics.

Finally, we also notice that occasionally the GARCH performs better than the Int-GARCH model in our experiments. For example, in Table \ref{tab:in-sample}, $\text{R}^2$ of GARCH is higher for BAC, and HMSE of GARCH is smaller for MSFT. In table \ref{tab:out-sample-1}, the Int-GARCH is dominated by GARCH for KO, MSFT, and T. Theoretically, the comparison between the GARCH and the Int-GARCH are complicated. Intuitively understanding, the advantage of our Int-GARCH model is that it makes full use of all possible values in the return range as volatility proxies and summarizes the information by integrating the GARCH mechanism with the random set theory. By comparison, the traditional GARCH model only considers use of the information of a single (closing-to-closing) return as the volatility proxy. Generally, the ``summary'' that accounts for more information from the data is expected to reflect the volatility more accurately. However, noting that different measures of model performance characterize different features of the model fitting the data, it is therefore possible occasionally that a simpler GARCH model with the 
point-valued return may outperform the more involved ``summary'' of Int-GARCH in terms of one performance measure. With regard to this, more extensive experiments with much more and larger data sets are needed to thoroughly investigate and optimize the performance of the Int-GARCH model in the near future.

\begin{table}[htbp]
  \centering
  \caption{In-sample comparisons of Int-GARCH (1,1,1) and GARCH (1,1). Larger values of $R^2$ and smaller values of QLIKE and HMSE indicate higher ranking of the volatility model.}
    \begin{tabular}{rlrrrrrr}
    \toprule
          &       & \multicolumn{2}{c}{$\textbf{R}^2$} & \multicolumn{2}{c}{\textbf{QLIKE}} & \multicolumn{2}{c}{\textbf{HMSE}} \\
          &       & \multicolumn{1}{c}{GARCH} & \multicolumn{1}{c}{Int-GARCH} & \multicolumn{1}{c}{GARCH} & \multicolumn{1}{c}{Int-GARCH} & \multicolumn{1}{c}{GARCH} & \multicolumn{1}{c}{Int-GARCH} \\
    \multicolumn{1}{l}{\textbf{Stocks}} & AAPL  & \textbf{0.4707} & 0.3822 & -7.0084 & \textbf{-7.0484} & 0.802 & \textbf{0.5178} \\
          & AXP   & 0.5099 & \textbf{0.5692} & -7.0235 & \textbf{-7.0851} & 0.36  & \textbf{0.3431} \\
          & BA    & 0.4949 & \textbf{0.5446} & -7.2523 & \textbf{-7.2592} & 0.5623 & \textbf{0.4291} \\
          & BAC   & \textbf{0.5179} & 0.4889 & -6.6973 & \textbf{-6.7652} & 1.104 & \textbf{0.665} \\
          & DD    & 0.4637 & \textbf{0.5778} & -7.29 & \textbf{-7.3093} & 0.5374 & \textbf{0.4178} \\
          & JPM   & 0.3728 & \textbf{0.4608} & -6.8266 & \textbf{-6.9093} & 0.6181 & \textbf{0.4411} \\
          & KO    & 0.4512 & \textbf{0.5287} & -8.2397 & \textbf{-8.242} & 0.5816 & \textbf{0.3866} \\
          & MSFT  & 0.3831 & \textbf{0.4323} & -7.4539 & \textbf{-7.4733} & \textbf{0.8664} & 1.1465 \\
          & T     & 0.4099 & \textbf{0.4522} & -7.5155 & \textbf{-7.5569} & 1.5514 & \textbf{0.7272} \\
          & WMT   & 0.3644 & \textbf{0.4009} & -7.8462 & \textbf{-7.8861} & 0.7694 & \textbf{0.7242} \\
          &       &       &       &       &       &       &  \\
    \multicolumn{1}{l}{\textbf{Indices}} & DJI   & \textbf{0.5006} & 0.4988 & -8.5019 & \textbf{-8.5877} & 0.5695 & \textbf{0.5428} \\
          & SPX   & 0.5209 & \textbf{0.5286} & -8.401 & \textbf{-8.4873} & \textbf{0.481} & 0.4839 \\
          & FTSE  & 0.4769 & \textbf{0.538} & -8.4324 & \textbf{-8.5158} & 0.4122 & \textbf{0.359} \\
          & CAC   & 0.4171 & \textbf{0.4988} & -8.1337 & \textbf{-8.1737} & 0.4158 & \textbf{0.3854} \\
    \bottomrule
    \end{tabular}%
  \label{tab:in-sample}%
\end{table}

\begin{table}[htbp]
  \centering
  \caption{Comparisons of 1-step-ahead predictions of Int-GARCH (1,1,1) and GARCH (1,1). }
    \begin{tabular}{rrrrrrrr}
    \toprule
    \multicolumn{8}{c}{\textbf{ 1-step-ahead Prediction}} \\
    \midrule
          &       & \multicolumn{2}{c}{$\textbf{R}^2$} & \multicolumn{2}{c}{\textbf{QLIKE}} & \multicolumn{2}{c}{\textbf{HMSE}} \\
          &       & \multicolumn{1}{c}{GARCH} & \multicolumn{1}{c}{Int-GARCH} & \multicolumn{1}{c}{GARCH} & \multicolumn{1}{c}{Int-GARCH} & \multicolumn{1}{c}{GARCH} & \multicolumn{1}{c}{Int-GARCH} \\
    \textbf{Stocks} & AAPL  & 0.0485 & \textbf{0.298} & -7.5694 & \textbf{-7.6838} & 0.5426 & \textbf{0.4262} \\
          & AXP   & 0.2591 & \textbf{0.5507} & -7.3522 & \textbf{-7.4218} & 0.3678 & \textbf{0.248} \\
          & BA    & 0.2909 & \textbf{0.4847} & -7.4518 & \textbf{-7.5054} & 0.3293 & \textbf{0.2911} \\
          & BAC   & 0.2461 & \textbf{0.4203} & -6.6467 & \textbf{-6.7399} & 0.4693 & \textbf{0.443} \\
          & DD    & 0.3053 & \textbf{0.5058} & -7.3962 & \textbf{-7.4448} & 0.3275 & \textbf{0.2614} \\
          & JPM   & 0.3587 & \textbf{0.4318} & -7.0733 & \textbf{-7.1648} & 0.392 & \textbf{0.3583} \\
          & KO    & 0.2844 & \textbf{0.4279} & -8.3893 & \textbf{-8.4008} & \textbf{0.3503} & 0.445 \\
          & MSFT  & 0.1032 & \textbf{0.1894} & -7.6106 & \textbf{-7.6581} & \textbf{0.7244} & 0.7286 \\
          & T     & 0.153 & \textbf{0.2121} & -8.1725 & \textbf{-8.1884} & \textbf{0.4431} & 0.6729 \\
          & WMT   & 0.2538 & \textbf{0.4033} & -6.8641 & \textbf{-8.3991} & 0.7947 & \textbf{0.3363} \\
          &       &       & \textbf{} &       & \textbf{} &       & \textbf{} \\
    \textbf{Indices} & DJI   & 0.3204 & \textbf{0.5334} & -8.5295 & \textbf{-8.6221} & 0.3875 & \textbf{0.3421} \\
          & SPX   & 0.3413 & \textbf{0.5882} & -8.3776 & \textbf{-8.461} & 0.3984 & \textbf{0.3436} \\
          & FTSE  & 0.3421 & \textbf{0.5555} & -8.3617 & \textbf{-8.3958} & 0.4161 & \textbf{0.3795} \\
          & CAC   & 0.3673 & \textbf{0.5285} & -7.8501 & \textbf{-7.8694} & 0.4095 & \textbf{0.3711} \\
    \bottomrule
    \end{tabular}%
  \label{tab:out-sample-1}%
\end{table}%

\begin{table}[htbp]
  \centering
  \caption{Comparisons of 2-step-ahead predictions of Int-GARCH (1,1,1) and GARCH (1,1).}
    \begin{tabular}{rrrrrrrr}
    \toprule
    \multicolumn{8}{c}{\textbf{2-step-ahead Prediction}} \\
    \midrule
          &       & \multicolumn{2}{c}{$\textbf{R}^2$} & \multicolumn{2}{c}{\textbf{QLIKE}} & \multicolumn{2}{c}{\textbf{HMSE}} \\
          &       & \multicolumn{1}{c}{GARCH} & \multicolumn{1}{c}{Int-GARCH} & \multicolumn{1}{c}{GARCH} & \multicolumn{1}{c}{Int-GARCH} & \multicolumn{1}{c}{GARCH} & \multicolumn{1}{c}{Int-GARCH} \\
    \textbf{Stocks} & AAPL  & 0.0132 & \textbf{0.1574} & -7.5383 & \textbf{-7.5879} & 0.6343 & \textbf{0.4554} \\
          & AXP   & 0.1275 & \textbf{0.2865} & -7.3106 & \textbf{-7.3747} & 0.635 & \textbf{0.3224} \\
          & BA    & 0.2132 & \textbf{0.3291} & -7.4326 & \textbf{-7.4655} & 0.3797 & \textbf{0.3143} \\
          & BAC   & 0.1367 & \textbf{0.1879} & -6.5818 & \textbf{-6.6704} & 0.8272 & \textbf{0.6055} \\
          & DD    & 0.2097 & \textbf{0.3287} & -7.3671 & \textbf{-7.4149} & 0.4624 & \textbf{0.2928} \\
          & JPM   & \textbf{0.2601} & 0.2432 & -7.0329 & \textbf{-7.1009} & \textbf{0.6407} & 0.6817 \\
          & KO    & 0.1908 & \textbf{0.2407} & \textbf{-8.3579} & -8.3372 & \textbf{0.4782} & 0.6853 \\
          & MSFT  & 0.0596 & \textbf{0.1203} & -7.5891 & \textbf{-7.6448} & 0.7986 & \textbf{0.5681} \\
          & T     & 0.0888 & \textbf{0.1128} & -8.1397 & \textbf{-8.1643} & 0.6592 & \textbf{0.4805} \\
          & WMT   & 0.1925 & \textbf{0.2499} & -6.7841 & \textbf{-8.3713} & 0.8143 & \textbf{0.347} \\
          &       &       & \textbf{} &       & \textbf{} &       & \textbf{} \\
    \textbf{Indices} & DJI   & 0.2046 & \textbf{0.3392} & -8.4939 & \textbf{-8.5731} & 0.525 & \textbf{0.3758} \\
          & SPX   & 0.2061 & \textbf{0.3736} & -8.3409 & \textbf{-8.4095} & 0.5378 & \textbf{0.3789} \\
          & FTSE  & 0.2581 & \textbf{0.385} & -8.3375 & \textbf{-8.3469} & 0.4742 & \textbf{0.4309} \\
          & CAC   & 0.2666 & \textbf{0.3509} & \textbf{-7.8209} & -7.8204 & 0.4894 & \textbf{0.4406} \\
    \bottomrule
    \end{tabular}%
  \label{tab:out-sample-2}%
\end{table}%

\begin{table}[htbp]
  \centering
  \caption{Comparisons of 5-step-ahead (one week) predictions of Int-GARCH (1,1,1) and GARCH (1,1). }
    \begin{tabular}{rrrrrrrr}
    \toprule
    \multicolumn{8}{c}{\textbf{5-step-ahead Prediction}} \\
    \midrule
          &       & \multicolumn{2}{c}{$\textbf{R}^2$} & \multicolumn{2}{c}{\textbf{QLIKE}} & \multicolumn{2}{c}{\textbf{HMSE}} \\
          &       & \multicolumn{1}{c}{GARCH} & \multicolumn{1}{c}{Int-GARCH} & \multicolumn{1}{c}{GARCH} & \multicolumn{1}{c}{Int-GARCH} & \multicolumn{1}{c}{GARCH} & \multicolumn{1}{c}{Int-GARCH} \\
    \textbf{Stocks} & AAPL  & 0.002 & \textbf{0.0141} & \textbf{-7.4968} & -7.4147 & 0.6861 & \textbf{0.5453} \\
          & AXP   & 0.0278 & \textbf{0.0488} & -7.2142 & \textbf{-7.2724} & 1.6016 & \textbf{0.5616} \\
          & BA    & 0.1281 & \textbf{0.1872} & -7.3887 & \textbf{-7.4081} & 0.5966 & \textbf{0.3523} \\
          & BAC   & \textbf{0.0625} & 0.0189 & -6.401 & \textbf{-6.4962} & 3.9111 & \textbf{1.4395} \\
          & DD    & 0.1072 & \textbf{0.1175} & -7.2938 & \textbf{-7.3466} & 0.946 & \textbf{0.4013} \\
          & JPM   & \textbf{0.1679} & 0.0986 & -6.9515 & \textbf{-7.0213} & 1.2665 & \textbf{0.7303} \\
          & KO    & \textbf{0.0875} & 0.0743 & -8.2977 & \textbf{-8.2983} & 1.2164 & \textbf{1.0582} \\
          & MSFT  & 0.0336 & \textbf{0.0416} & -7.551 & \textbf{-7.5808} & 1.0182 & \textbf{0.7925} \\
          & T     & 0.052 & \textbf{0.0717} & \textbf{-8.0929} & -8.0853 & 0.95  & \textbf{0.4313} \\
          & WMT   & \textbf{0.1068} & 0.0463 & -6.5603 & \textbf{-8.3041} & 0.8564 & \textbf{0.4626} \\
          &       &       & \textbf{} &       & \textbf{} &       & \textbf{} \\
    \textbf{Indices} & DJI   & 0.1092 & \textbf{0.1266} & -8.4081 & \textbf{-8.4504} & 1.1346 & \textbf{0.6644} \\
          & SPX   & 0.1083 & \textbf{0.1302} & -8.2571 & \textbf{-8.2834} & 1.1027 & \textbf{0.6043} \\
          & FTSE  & 0.1334 & \textbf{0.138} & -8.2455 & \textbf{-8.2624} & 1.154 & \textbf{0.5915} \\
          & CAC   & 0.1696 & \textbf{0.1709} & -7.7553 & \textbf{-7.7696} & 0.7602 & \textbf{0.5749} \\
    \bottomrule
    \end{tabular}%
  \label{tab:out-sample-5}%
\end{table}%

\section{Conclusion}\label{sec:conclude}
In this paper, we have developed an interval-valued GARCH model for analyzing range-measured return processes. It can be viewed as an extension of the point-valued GARCH model that allows for interval-valued returns to produce ``information-richer'' estimation of the volatility. Inferences of our Int-GARCH model can be made based on the maximum likelihood method, which has been shown to have both consistency and asymptotic normality. Our empirical study of stocks and indices data has demonstrated the advantages of Int-GARCH model over the GARCH for both in-sample and out-of-sample performances. 

An equally important contribution of our Int-GARCH model lies in that it forms a pioneer study in the area of conditional heteroscedasticity for interval-valued time series. Further improvements to the model can be made by incorporating interactions between the level and the range. K\"orner and N\"ather (\cite{Korner01}) proposed for a more general space a generalized $L_2$ metric, which when restricted to $\mathcal{K}_{\mathcal{C}}(\mathbb{R})$ is
\begin{equation*}
\rho^2_K(x, y)=\sum_{(u, v)\in S^{0}\times S^{0}}\left(s_{x}(u)-s_{y}(u)\right)\left(s_{x}(v)-s_{y}(v)\right)K(u, v),\ \ 
x,y\in\mathcal{K}_{\mathcal{C}}(\mathbb{R}),
\end{equation*} 
where $K$ is a symmetric positive definite kernel. It can be represented by the center-radius form as 
\begin{equation}\label{def:k2}
  \rho^2_K(x, y)=A_{11}(x^C-y^C)^2+A_{22}(x^R-y^R)^2+2A_{12}(x^C-y^C)(x^R-y^R),
\end{equation}
where $A$ as a function of $K$ is a symmetric positive definite matrix. So $\rho_K$ is a further generalization of $\rho_W$ in (\ref{def:w2}) that takes into account the correlation between the center and the radius, and can be utilized to model the level-range interaction. 




\clearpage
\section{Appendix}
\subsection{Proof of Theorem \ref{thm:mean-gen}}
\begin{proof}
By the definition of (\ref{def:X_it}), (\ref{igarch_5}) can be rewritten as
\begin{eqnarray*}
h_t&=&\mu+\sum_{i=1}^{p}\alpha_i|\epsilon_{t-i}|h_{t-i}+\sum_{i=1}^{q}\beta_i\eta_{t-i}h_{t-i}+\sum_{i=1}^{w}\gamma_ih_{t-i}\\
&=&\mu+\sum_{i=1}^{k}x_{i,t-i}h_{t-i}.
\end{eqnarray*}
Expanding $h_t$ recursively, we obtain
\begin{eqnarray}
h_t&=&\mu+\sum_{i=1}^{k}x_{i,t-i}\left(\mu+\sum_{j=1}^{k}x_{j,t-i-j}h_{t-i-j}\right)\nonumber\\
&=&\mu\left(1+\sum_{i=1}^{k}x_{i,t-i}\right)+\sum_{i=1}^{k}\sum_{j=1}^{k}x_{i,t-i}x_{j,t-i-j}h_{t-i-j}\nonumber\\
&=&\cdots\nonumber\\
&=&\mu\left[1+\sum_{n=1}^{N}\sum_{i_1=1}^{k}\cdots\sum_{i_n=1}^{k}\left(\prod_{j=1}^{n}x_{i_j, t-i_1-\cdots-i_j}\right)\right]\nonumber\\
&&+\sum_{i_1=1}^{k}\cdots\sum_{i_n{N+1}=1}^{k}\left(\prod_{j=1}^{N+1}x_{i_j, t-i_1-\cdots-i_j}\right)h_{t-i_1-\cdots-i_{N+1}}.\label{exp-ht}
\end{eqnarray}
Notice that $x_{i,t}$ and $x_{j,s}$ are independent, $\forall i,j\in\mathbb{N}$ and $t\neq s$. Taking expectations on both sizes of (\ref{exp-ht}), we get
\begin{eqnarray*}
E\left(h_t\right)&=&\mu\left[1+\sum_{n=1}^{N}\sum_{i_1=1}^{k}\cdots\sum_{i_n=1}^{k}\left(\prod_{j=1}^{n}\mu_{i_j}\right)\right]\\
&&+\sum_{i_1=1}^{k}\cdots\sum_{i_{N+1}=1}^{k}\left(\prod_{j=1}^{N+1}\mu_{i_j}\right)E\left(h_{t-i_1-\cdots-i_{N+1}}\right)\\
&&:=I+II. 
\end{eqnarray*}
The first term
\begin{eqnarray*}
  I=\mu\left[1+\sum_{n=1}^{N}\sum_{i_1=1}^{k}\cdots\sum_{i_n=1}^{k}\left(\prod_{j=1}^{n}\mu_{i_j}\right)\right]
  =\mu\left[1+\sum_{n=1}^{N}\left(\sum_{j=1}^{k}\mu_{j}\right)^n\right],\ \ \forall\ N\in\mathbb{N}.
\end{eqnarray*}
The second term 
\begin{eqnarray*}
  II&=&\sum_{i_1=1}^{k}\cdots\sum_{i_{N+1}=1}^{k}\left(\prod_{j=1}^{N+1}\mu_{i_j}\right)E\left(h_{t-i_1-\cdots-i_{N+1}}\right)\\
  &\leq&\sum_{i_1=1}^{k}\cdots\sum_{i_{N+1}=1}^{k}\left(\prod_{j=1}^{N+1}\mu_{i_j}\right)\max\left\{E\left(h_{t-l}\right): N+1\leq l\leq k(N+1)\right\}\\
  &=&\max\left\{E\left(h_{t-l}\right): N+1\leq l\leq k(N+1)\right\}\left(\sum_{j=1}^{k}\mu_{j}\right)^{N+1}, \ \ \forall\ N\in\mathbb{N}.
\end{eqnarray*}
Therefore, $E\left(h_t\right)<\infty$ if and only if $\sum_{i=1}^{k}\mu_i<1$. When it is satisfied, 
\begin{eqnarray*}
E\left(h_t\right)&=&\lim_{N\to\infty}\mu\left[1+\sum_{n=1}^{N}\left(\sum_{j=1}^{k}\mu_{j}\right)^n\right]\\
&&+\lim_{N\to\infty}\max\left\{E\left(h_{t-l}\right): N+1\leq l\leq k(N+1)\right\}\left(\sum_{j=1}^{k}\mu_{j}\right)^{N+1}\\
&=&\frac{\mu}{1-\sum_{i=1}^{k}\mu_i},
\end{eqnarray*}
by the finiteness of $E\left(h_{-\infty}\right)$. The formula for $E\left(r_t\right)$ follows immediately from the Aumann expectation.
\end{proof}

\subsection{Proof of Theorem \ref{thm:mean}}
\begin{proof}
By recursive calculations, 
\begin{eqnarray*}
h_{t} 
 & = & \mu\left(1+\sum_{i=1}^{N}\prod_{j=1}^{i}x_{t-j}\right)+h_{t-(N+1)}\prod_{j=1}^{N+1}x_{t-j},\qquad\forall N\in\mathbb{N}.
\end{eqnarray*}
Taking expectations on both sides, we get
\begin{eqnarray*}
Eh_{t} & = & E\left(\mu\left(1+\sum_{i=1}^{N}\prod_{j=1}^{i}x_{t-j}\right)+h_{t-(N+1)}\prod_{j=1}^{N+1}x_{t-j}\right)\\
 & = & \mu\sum_{i=0}^{N}\left(Ex_{t}\right)^{i}+\left(Ex_{t}\right)^{N+1}Eh_{t-(N+1)},
\end{eqnarray*}
for all $N\in\mathbb{N}$. Letting $N\to\infty$, 
\begin{eqnarray*}
Eh_{t} & = & \mu\sum_{i=0}^{\infty}\left(Ex_{t}\right)^{i}+\lim_{N\rightarrow\infty}\left(Ex_{t}\right)^{N+1}Eh_{t-(N+1)}\\
 & = & \mu\sum_{i=0}^{\infty}\left(Ex_{t}\right)^{i}\\
 & = &  \dfrac{\mu}{1-E\left(x_{t}\right)}\\
 & = & \dfrac{\mu}{1-\alpha_{1}\sqrt{2/\pi}-\beta_{1}k-\gamma_{1}},
\end{eqnarray*}
since $|Ex_t|<1$. On the other hand, if $Ex_t\geq 1$, then $Eh_t>\mu\sum_{i=0}^{\infty}(Ex_t)^i=\infty$. Therefore, $Eh_t<\infty$ if and only if
\begin{eqnarray*}
\left|Ex_{t}\right| & = & \left|E\left(\alpha_{1}\left|\varepsilon_{t-1}\right|+\beta_{1}\eta_{t-1}+\gamma_{1}\right)\right|\\
 & = & \left|\alpha_{1}\sqrt{2/\pi}+\beta_{1}k+\gamma_{1}\right|\\
 & < & 1,
\end{eqnarray*}
and when this is satisfied, the Aumann expectation of $r_{t}$ is found to be
\begin{eqnarray*}
Er_{t} 
 & = & E\left[h_{t}\left(\varepsilon_{t}-\eta_{t}\right),h_{t}\left(\varepsilon_{t}+\eta_{t}\right)\right]\\
 & = & \left[E\left(h_{t}\right)E\left(\varepsilon_{t}-\eta_{t}\right), E\left(h_{t}\right)E\left(\varepsilon_{t}+\eta_{t}\right)\right]\\
 & = & \left[-kE\left(h_{t}\right),kE\left(h_{t}\right)\right].
\end{eqnarray*}
\end{proof}

\subsection{Proof of Theorem \ref{thm:var}}
\begin{proof}
Recall, from the proof of Theorem \ref{thm:mean}, that 
$$h_{t}=\mu\left(1+\sum_{i=1}^{N}\prod_{j=1}^{i}x_{t-j}\right)+h_{t-(N+1)}\prod_{j=1}^{N+1}x_{t-j},\ \forall N\in\mathbb{N}.$$
Consequently, 
\begin{eqnarray*}
h_{t}^{2} & = & \left(\mu\left(1+\sum_{i=1}^{N}\prod_{j=1}^{i}x_{t-j}\right)+h_{t-(N+1)}\prod_{j=1}^{N+1}x_{t-j}\right)^{2}\\
 & = & \mu^{2}\left(1+\sum_{i=1}^{N}\prod_{j=1}^{i}x_{t-j}\right)^{2}+2\mu\left(1+\sum_{i=1}^{N}\prod_{j=1}^{i}x_{t-j}\right)\left(h_{t-(N+1)}\prod_{j=1}^{N+1}x_{t-j}\right)\\
 &  & +\left(h_{t-(N+1)}\prod_{j=1}^{N+1}x_{t-j}\right)^{2}\\
 & = & \mu^{2}\left(1+\sum_{i=1}^{N}\prod_{j=1}^{i}x_{t-j}\right)^{2}+2\mu h_{t-(N+1)}\prod_{j=1}^{N+1}x_{t-j}\\
 &  & +2\mu h_{t-(N+1)}\prod_{j=1}^{N+1}x_{t-j}\left(\sum_{i=1}^{N}\prod_{j=1}^{i}x_{t-j}\right)+h_{t-(N+1)}^{2}\prod_{j=1}^{N+1}x_{t-j}^{2},
 \forall N\in\mathbb{N}.
\end{eqnarray*}
Note that
\begin{eqnarray*}
 &  & \left(1+\sum_{i=1}^{N}\prod_{j=1}^{i}x_{t-j}\right)^{2}\\
 & = & \left(1+x_{t-1}+x_{t-1}x_{t-2}+\cdots+x_{t-1}x_{t-2}\cdots x_{t-N}\right)^{2}\\
 & = & 1+\sum_{i=1}^{N}\prod_{j=1}^{i}x_{t-j}^{2}+2\sum_{i=1}^{N}\prod_{j=1}^{i}x_{t-j}+2\sum_{i=1}^{N-1}\left[\left(\prod_{j=1}^{i}x_{t-j}^{2}\right)\left(\sum_{k=i+1}^{N}\,\prod_{l=i+1}^{k}x_{t-l}\right)\right],
\end{eqnarray*}
and
\begin{equation*}
  \,\,\,\,\prod_{j=1}^{N+1}x_{t-j}\left(\sum_{i=1}^{N}\prod_{j=1}^{i}x_{t-j}\right)
  =\sum_{i=1}^{N}\left[\left(\prod_{j=1}^{i}x_{t-j}^{2}\right)\left(\prod_{k=i+1}^{N+1}x_{t-k}\right)\right].
\end{equation*}
Therefore, 
\begin{eqnarray*}
h_{t}^{2} & = & \mu^{2}\left\{1+\sum_{i=1}^{N}\left(\prod_{j=1}^{i}x_{t-j}^{2}+2\prod_{j=1}^{i}x_{t-j}\right)+2\sum_{i=1}^{N-1}\left[\left(\prod_{j=1}^{i}x_{t-j}^{2}\right)\left(\sum_{k=i+1}^{N}\,\prod_{l=i+1}^{k}x_{t-l}\right)\right]\right\}\\
 &  & +2\mu h_{t-(N+1)}\prod_{j=1}^{N+1}x_{t-j}+2\mu h_{t-(N+1)}\sum_{i=1}^{N}\left[\left(\prod_{j=1`}^{i}x_{t-j}^{2}\right)\left(\prod_{k=i+1}^{N+1}x_{t-k}\right)\right]\\
 &  & +h_{t-(N+1)}^{2}\prod_{j=1}^{N+1}x_{t-j}^{2}\\
 & = & \mu^{2}\left\{1+\sum_{i=1}^{N}\left(\prod_{j=1}^{i}x_{t-j}^{2}+2\prod_{j=1}^{i}x_{t-j}\right)+2\sum_{i=1}^{N-1}\left[\left(\prod_{j=1}^{i}x_{t-j}^{2}\right)\left(\sum_{k=i+1}^{N}\,\prod_{l=i+1}^{k}x_{t-l}\right)\right]\right\}\\
 &  & +2\mu h_{t-(N+1)}\left(\prod_{j=1}^{N+1}x_{t-j}+\sum_{i=1}^{N}\left[\left(\prod_{j=1`}^{i}x_{t-j}^{2}\right)\left(\prod_{k=i+1}^{N+1}x_{t-k}\right)\right]\right)+h_{t-(N+1)}^{2}\prod_{j=1}^{N+1}x_{t-j}^{2},\\
 & & \forall N\in\mathbb{N}.
\end{eqnarray*}
Let $C_{1}=E\left(x_{t}\right)$ and $C_{2}=E\left(x_{t}^{2}\right)$.
By the independence of $h_{t-(N+1)}$ and $x_{t-j}$, $\forall j\leq N+1$ and the fact that $\left\{ x_{t}\right\} $ are iid, we obtain,
\begin{eqnarray*}
E\left(h_{t}^{2}\right) & = & E\left\{\mu^{2}\left(1+\sum_{i=1}^{N}\left(\prod_{j=1}^{i}x_{t-j}^{2}+2\prod_{j=1}^{i}x_{t-j}\right)+2\sum_{i=1}^{N-1}\left[\left(\prod_{j=1}^{i}x_{t-j}^{2}\right)\left(\sum_{k=i+1}^{N}\,\prod_{l=i+1}^{k}x_{t-l}\right)\right]\right)\right.\\
 &  & \left.+2\mu h_{t-(N+1)}\left(\prod_{j=1}^{N+1}x_{t-j}+\sum_{i=1}^{N}\left[\left(\prod_{j=1`}^{i}x_{t-j}^{2}\right)\left(\prod_{k=i+1}^{N+1}x_{t-k}\right)\right]\right)+h_{t-(N+1)}^{2}\prod_{j=1}^{N+1}x_{t-j}^{2}\right\}\\
 & = & \mu^{2}+\mu^{2}\sum_{i=1}^{N}E\left[\prod_{j=1}^{i}x_{t-j}^{2}\right]+2\mu^{2}\sum_{i=1}^{N}E\left[\prod_{j=1}^{i}x_{t-j}\right]\\
 &  & +2\mu^{2}\sum_{i=1}^{N-1}E\left[\left(\prod_{j=1}^{i}x_{t-j}^{2}\right)\left(\sum_{k=i+1}^{N}\,\prod_{l=i+1}^{k}x_{t-l}\right)\right]\\
 &  & +2\mu E\left[h_{t-(N+1)}\right]\cdot\left(E\left[\prod_{j=1}^{N+1}x_{t-j}\right]+\sum_{i=1}^{N}E\left[\left(\prod_{j=1}^{i}x_{t-j}^{2}\right)\left(\prod_{k=i+1}^{N+1}x_{t-k}\right)\right]\right)\\
 &  & +E\left[h_{t-(N+1)}^{2}\right]\cdot E\left[\prod_{j=1}^{N+1}x_{t-j}^{2}\right]\\
 & = & \mu^{2}+\mu^{2}\sum_{i=1}^{N}C_{2}^{i}+2\mu^{2}\sum_{i=1}^{N}C_{1}^{i}+2\mu^{2}\sum_{i=1}^{N-1}\left[C_{2}^{i}\left(\sum_{k=i+1}^{N}C_{1}^{k-i}\right)\right]\\
 &  & +2\mu E\left[h_{t-(N+1)}\right]\cdot\left(C_{1}^{N+1}+\sum_{i=1}^{N}C_{2}^{i}C_{1}^{N-i+1}\right)+E\left[h_{t-(N+1)}^{2}\right]\cdot C_{2}^{N+1},\\
 &  & \forall N\in\mathbb{N}.
\end{eqnarray*}
The geometric series $\sum_{i=1}^{N}C_{2}^{i}$ and $\sum_{i=1}^{N}C_{1}^{i}$ converge if and only if $\left|C_{1}\right|<1$ and $\left|C_{2}\right|<1$. Since
\begin{equation*}
  0<C_1=E(x_t)\leq\sqrt{E(x_t^2)}=\sqrt{C_2},
\end{equation*}
only $|C_2|<1$ is sufficient here. Under this assumption, the fourth term in the above expression, 
\begin{eqnarray*}
&&\sum_{i=1}^{N-1}\left[C_{2}^{i}\left(\sum_{k=i+1}^{N}C_{1}^{k-i}\right)\right]\\ 
 & = & \sum_{i=1}^{N-1}\left[C_{2}^{i}\left(\dfrac{C_{1}\left(C_{1}^{N-i}-1\right)}{C_{1}-1}\right)\right]
  = \dfrac{C_{1}}{C_{1}-1}\sum_{i=1}^{N-1}\left[C_{2}^{i}C_{1}^{N-i}-C_{2}^{i}\right]\\
 & = & \dfrac{C_{1}}{C_{1}-1}\left[C_{1}^{N}\cdot\sum_{i=1}^{N-1}\left(\dfrac{C_{2}}{C_{1}}\right)^{i}-\sum_{i=1}^{N-1}C_{2}^{i}\right]
 = \dfrac{C_{1}}{C_{1}-1}\left[C_1^N\cdot\dfrac{\left(\dfrac{C_2}{C_1}\right)-\left(\dfrac{C_2}{C_1}\right)^N}{1-\dfrac{C_2}{C_1}}-\sum_{i=1}^{N-1}C_{2}^{i}\right]\\
 &=& \dfrac{C_1}{C_1-1}\left[\dfrac{\left(\dfrac{C_2}{C_1}\right)C_1^N-C_2^N}{1-\dfrac{C_2}{C_1}}-\sum_{i=1}^{N-1}C_{2}^{i}\right]<\infty,\ \ N\to\infty. 
\end{eqnarray*}
And similarly,
\begin{eqnarray*}
\sum_{i=1}^{N}C_{2}^{i}C_{1}^{N-i+1} & = & C_{1}^{N+1}\sum_{i=1}^{N}\left(\dfrac{C_{2}}{C_{1}}\right)^{i}<\infty,\ \ N\to\infty. 
\end{eqnarray*}
Therefore, $E\left(h_t^2\right)<\infty$ if and only if $|C_2|<1$, or $C_2<1$, since it is positive. Under this assumption and in addition that $Eh^2_{-\infty}<\infty$, letting $N\to\infty$, we find the second moment of $h_t$ to be
\begin{eqnarray*}
E\left(h_{t}^{2}\right) 
 & = & \mu^{2}\left(1-\dfrac{C_{2}}{C_{2}-1}-2\dfrac{C_{1}}{C_{1}-1}+2C_{1}C_{2}\dfrac{1}{\left(C_{2}-1\right)\left(C_{1}-1\right)}\right)\\
 & = & \mu^{2}\left(\dfrac{C_{1}C_{2}-C_{1}-C_{2}+1-C_{2}C_{1}+C_{2}-2C_{1}C_{2}+2C_{1}+2C_{1}C_{2}}{\left(C_{2}-1\right)\left(C_{1}-1\right)}\right)\\
 & = & \mu^{2}\dfrac{C_{1}+1}{\left(C_{2}-1\right)\left(C_{1}-1\right)}.
\end{eqnarray*}
Consequently, the unconditional variance of $r_t$ is
\begin{eqnarray*}
\mbox{Var}(r_{t})  & = & \mbox{Var}\left(h_{t}\varepsilon_{t}\right)+\mbox{Var}\left(h_{t}\eta_{t}\right)\\
 & = & \left(1+k+k^{2}\right)E\left(h_{t}^{2}\right)-k^{2}\left[E\left(h_{t}\right)\right]^{2}.
\end{eqnarray*}
This completes the proof.
\end{proof}

\subsection{Proof of Theorem \ref{thm:cov}}
\begin{proof}
First we notice, $\forall t, s\in\mathbb{N}$,
\begin{eqnarray*}
  r_{t} &= &\left[h_{t}\varepsilon_{t}-h_{t}\eta_{t},h_{t}\varepsilon_{t}+h_{t}\eta_{t}\right],\\
  r_{t+s} &=&\left[h_{t+s}\varepsilon_{t+s}-h_{t+s}\eta_{t+s},h_{t+s}\varepsilon_{t+s}+h_{t+s}\eta_{t+s}\right],
\end{eqnarray*}
and therefore,
\begin{equation}\label{eqn1}
\mbox{Cov}\left(r_{t},r_{t+s}\right)=\mbox{Cov}\left(h_{t}\varepsilon_{t},h_{t+s}\varepsilon_{t+s}\right)+\mbox{Cov}\left(h_{t}\eta_{t},h_{t+s}\eta_{t+s}\right).
\end{equation}
The first term
\begin{eqnarray}
\mbox{Cov}\left(h_{t}\varepsilon_{t},h_{t+s}\varepsilon_{t+s}\right) 
& = & E\left[\left(h_{t}\varepsilon_{t}-E\left(h_{t}\epsilon_{t}\right)\right)\left(h_{t+s}\varepsilon_{t+s}-E\left(h_{t+s}\epsilon_{t+s}\right)\right)\right]\nonumber\\
 & = & E\left(h_{t}h_{t+s}\varepsilon_{t}\cdot\varepsilon_{t+s}\right)\nonumber\\
 & = & \begin{cases}
E\left(h_{t}^{2}\varepsilon_{t}^{2}\right),\,\, & s=0\\
E\left(h_{t}h_{t+s}\varepsilon_{t}\right)\cdot E\left(\varepsilon_{t+s}\right),\,\, & |s|>0
\end{cases}\nonumber\\
 & = & \begin{cases}
E\left(\varepsilon_{t}^{2}\right)\cdot E\left(h_{t}^{2}\right),\,\, & s=0\\
0, & |s|>0
\end{cases}\nonumber\\
 & = & \begin{cases}\label{eqn2}
E\left(h_{t}^{2}\right),\,\, & s=0\\
0, & |s|>0
\end{cases}
\end{eqnarray}
since $\left\{ \varepsilon_{t}\right\} $ are i.i.d. \\

Similarly, the second term becomes
\begin{eqnarray}
 &  & \mbox{Cov}\left(h_{t}\eta_{t},h_{t+s}\eta_{t+s}\right)\nonumber\\
 & = & E\left[\left(h_{t}\eta_{t}-kE\left(h_{t}\right)\right)\left(h_{t+s}\eta_{t+s}-kE\left(h_{t+s}\right)\right)\right]\nonumber\\
 & = & E\left(h_{t}h_{t+s}\eta_{t}\eta_{t+s}\right)-kE\left(h_{t}\right)E\left(h_{t+s}\eta_{t+s}\right)-kE\left(h_{t+s}\right)E\left(h_{t}\eta_{t}\right)+k^{2}E\left(h_{t}\right)E\left(h_{t+s}\right)\nonumber\\
 & = & E\left(h_{t}h_{t+s}\eta_{t}\eta_{t+s}\right)-kE\left(h_{t}\right)kE\left(h_{t+s}\right)-kE\left(h_{t+s}\right)kE\left(h_{t}\right)+k^{2}E\left(h_{t}\right)E\left(h_{t+s}\right)
 \nonumber\\
 & = & E\left(h_{t}h_{t+s}\eta_{t}\eta_{t+s}\right)-k^{2}E\left(h_{t}\right)E\left(h_{t+s}\right)\nonumber\\
 & = & \begin{cases}
E\left(h_{t}^{2}\eta_{t}^{2}\right)-k^{2}\left[E\left(h_{t}\right)\right]^{2}, & s=0\\
E\left(h_{t}h_{t+s}\eta_{t}\right)E\left(\eta_{t+s}\right)-k^{2}E\left(h_{t}\right)E\left(h_{t+s}\right), & |s|>0
\end{cases}\nonumber\\
 & = & \begin{cases}\label{eqn3}
\left(k+k^{2}\right)E\left(h_{t}^{2}\right)-k^{2}\left[E\left(h_{t}\right)\right]^{2}, & s=0\\
kE\left(h_{t}h_{t+s}\eta_{t}\right)-k^{2}\left[E\left(h_{t}\right)\right]^2, & |s|>0.
\end{cases}
\end{eqnarray}
Plugging (\ref{eqn2}) and (\ref{eqn3}) into (\ref{eqn1}), we obtain
\begin{align*}
\mbox{Cov}\left(r_{t},r_{t+s}\right) 
 & =\begin{cases}
\left(1+k+k^{2}\right)E\left(h_{t}^{2}\right)-k^{2}\left[E\left(h_{t}\right)\right]^{2}, & s=0\\
kE\left(h_{t}h_{t+s}\eta_{t}\right)-k^{2}\left[E\left(h_{t}\right)\right]^{2}, & |s|>0, 
\end{cases}
\end{align*}
where $E\left(h_{t}\right)=\dfrac{\mu}{1-C_{1}}$, $E\left(h_{t}^{2}\right)=\mu^{2}\dfrac{C_{1}+1}{\left(C_{2}-1\right)\left(C_{1}-1\right)}$,
and 
\begin{equation*}
  E\left(h_{t}h_{t+s}\eta_{t}\right)
  =\dfrac{\mu^{2}k}{C_{1}-1}\left(-\dfrac{C_{1}^{s}-1}{C_{1}-1}+\dfrac{C_{1}^{s}+C_{1}^{s-1}}{C_{2}-1}\cdot\left[\alpha_{1}\sqrt{\dfrac{2}{\pi}}
  +  \beta_{1}\left(1+k\right)+\gamma_{1}\right]\right).
\end{equation*}
(see Lemma~\ref{eta-X and h-h-eta}, Section \ref{lemmas}).
\end{proof}

\subsection{Proof of Theorem \ref{thm:asymp-mle}}
\begin{proof}
Let $l_t(\bdsm{\theta})=-(k+1)\log(h_t)-\frac{\lambda_t^2}{2h_t^2}-\frac{\delta_t}{h_t}$, that is, $l(\bdsm{\theta})=\sum\limits_{t=1}^{T}l_t(\bdsm{\theta})$. The elements in 
$\nabla l_t$ are calculated to be
\begin{eqnarray*}
  &&\frac{\partial l}{\partial\mu}=\sum_{t}\frac{\partial l_t}{\partial h_t}
  =\sum_{t}\left(-\frac{k+1}{h_t}+\frac{\lambda_t^2}{h_t^3}+\frac{\delta_t}{h_t^2}\right),\\
  &&\frac{\partial l}{\partial\alpha_i}=\sum_{t}\frac{\partial l_t}{\partial h_t}|\lambda_{t-i}|,\ \ \ \ \ 
  \frac{\partial l}{\partial\beta_i}=\sum_{t}\frac{\partial l_t}{\partial h_t}\delta_{t-i},\ \ \ \ \ 
  \frac{\partial l}{\partial\gamma_i}=\sum_{t}\frac{\partial l_t}{\partial h_t}h_{t-i}.
\end{eqnarray*}
In addition, the elements in $\nabla^2l_t$ are 
\begin{eqnarray*}
  \frac{\partial^2l_t}{\partial\bdsm{\theta}_i\partial\bdsm{\theta}_j}
  =\frac{\partial^2l_t}{\partial h_t^2}\frac{\partial h_t}{\partial\bdsm{\theta}_i}\frac{\partial h_t}{\partial\bdsm{\theta}_j}
\end{eqnarray*}
(i) To prove consistency, we follow Weiss (1986) to verify the three conditions in Basawa et al. (1976) that guarantees the existence of a consistent root of the equation 
$\partial L(\bdsm{\theta})/\partial\bdsm{\theta}=0$. \\
(1) $\frac{1}{T}\sum_{t}\nabla l_t\left(\bdsm{\theta}_0\right)\stackrel{\mathcal{P}}{\to}0$ as $n\to\infty$;\\
(2) There exists a nonrandom positive definite matrix $M\left(\bdsm{\theta}_0\right)0$ such that $\forall\epsilon>0$:
$$P\left(-\frac{1}{T}\sum_{t}\nabla^2l_t\left(\bdsm{\theta}_0\right)\geq M\left(\bdsm{\theta}_0\right)\right)>1-\epsilon,\ \ 
\forall T>T_1(\epsilon);$$
(3) There exists a constant $M<\infty$ such that 
$$E\left\|\nabla_{ijk}^3l_t\left(\bdsm{\theta}\right)\right\|<M,\ \ 
\forall\bdsm{\theta}\in\Theta.$$
For (1), notice that evaluated at $\bdsm{\theta}_0$, 
\begin{equation*}
  E\left(\frac{\partial l_t}{\partial h_t}\right)
  =E\left[E\left(\frac{\partial l_t}{\partial h_t}|\mathcal{F}_{t-1}\right)\right]
  =E\left(-\frac{k+1}{h_t}+\frac{h_t^2}{h_t^3}+\frac{kh_t}{h_t^2}\right)
  =0. 
\end{equation*}
Thus, 
\begin{eqnarray*}
  &&E\left(\frac{\partial l_t}{\partial\mu}\right)=E\left(\frac{\partial l_t}{\partial h_t}\right)=0,\\
  &&E\left(\frac{\partial l_t}{\partial\alpha_i}\right)=E\left(\frac{\partial l_t}{\partial h_t}|\lambda_{t-i}|\right)
  =E\left[|\lambda_{t-i}|E\left(\frac{\partial l_t}{\partial h_t}|\mathcal{F}_{t-1}\right)\right]=0,\\
  &&E\left(\frac{\partial l_t}{\partial\beta_i}\right)=E\left(\frac{\partial l_t}{\partial h_t}\delta_{t-i}\right)
  =E\left[\delta_{t-i}E\left(\frac{\partial l_t}{\partial h_t}|\mathcal{F}_{t-1}\right)\right]=0,\\
  &&E\left(\frac{\partial l_t}{\partial\gamma_i}\right)=E\left(\frac{\partial l_t}{\partial h_t}h_{t-i}\right)
  =E\left[h_{t-i}E\left(\frac{\partial l_t}{\partial h_t}|\mathcal{F}_{t-1}\right)\right]=0.
\end{eqnarray*}
Namely, $E\left[\nabla l_t(\bdsm{\theta}_0)\right]=0$. By the ergodic theorem, 
$\frac{1}{T}\sum_{t=1}^{T}\nabla l_t(\bdsm{\theta}_0)\stackrel{\mathcal{P}}{\to}0$, $T\to\infty$. \\
For (2), in view of the ergodic theorem, it is sufficient to verify $E\left[-\nabla^2l_t(\bdsm{\theta}_0)\right]<\infty$. Notice
\begin{eqnarray*}
  \frac{\partial^2l_t}{\partial h_t^2}
  =\frac{k+1}{h_t^2}-\frac{3\lambda_t^2}{h_t^4}-\frac{2\delta_t}{h_t^3}.
\end{eqnarray*}
Therefore, 
\begin{eqnarray*}
  E\left[\frac{\partial^2l_t}{\partial\bdsm{\theta}_i\partial\bdsm{\theta}_j}(\bdsm{\theta}_0)\right]
  &=&E\left[E\left(\frac{\partial^2l_t}{\partial\bdsm{\theta}_i\partial\bdsm{\theta}_j}(\bdsm{\theta}_0)\right)|\mathcal{F}_{t-1}\right]\\
  &=&E\left[\left(\frac{k+1}{h_t^2}-\frac{3h_t^2}{h_t^4}-\frac{2kh_t}{h_t^3}\right)
  \frac{\partial h_t}{\partial\bdsm{\theta}_i}\frac{\partial h_t}{\partial\bdsm{\theta}_j}(\bdsm{\theta}_0)\right]\\
  &=&-(k+2)E\left[h_t^{-2}\frac{\partial h_t}{\partial\bdsm{\theta}_i}\frac{\partial h_t}{\partial\bdsm{\theta}_j}(\bdsm{\theta}_0)\right].
\end{eqnarray*}
Since $h_t^{-1}\frac{\partial h_t}{\partial\bdsm{\theta}_i}(\bdsm{\theta}_0)$ is bounded by $\max\left\{\frac{1}{\alpha_i}, \frac{1}{\beta_j}, \frac{1}{\gamma_k},
\frac{1}{\mu}\right\}$, we have that $E\left[-\nabla^2l_t(\bdsm{\theta}_0)\right]<\infty$, and hence condition (2) is satisfied.\\
For (3), we note that
\begin{eqnarray*}
  \frac{\partial^3l_t}{\partial\bdsm{\theta}_i\partial\bdsm{\theta}_j\partial\bdsm{\theta}_k}
 &=&\frac{\partial^3l_t}{\partial h_t^3}\frac{\partial h_t}{\partial\bdsm{\theta}_i}\frac{\partial h_t}{\partial\bdsm{\theta}_j}\frac{\partial h_t}{\partial\bdsm{\theta}_k}\\
 &=&-2\left[(k+1)-6h_t^{-2}\lambda_t^2-3h_t^{-1}\delta_t\right]h_t^{-3}\frac{\partial h_t}{\partial\bdsm{\theta}_i}
 \frac{\partial h_t}{\partial\bdsm{\theta}_j}\frac{\partial h_t}{\partial\bdsm{\theta}_k}
\end{eqnarray*} 
By the compactness of $\Theta$, $h_t^{-1}\frac{\partial h_t}{\partial\bdsm{\theta}_i}$ is uniformly bounded, and so is 
$h_t^{-3}\frac{\partial h_t}{\partial\bdsm{\theta}_i}\frac{\partial h_t}{\partial\bdsm{\theta}_j}\frac{\partial h_t}{\partial\bdsm{\theta}_k}$. Similarly, 
$h_t^{-1}, h_t^{-2}$ are also uniformly bounded since $h_t>\mu$. So it suffices to have $E\left(\lambda_t^2\right)<\infty$, which is equivalent to 
$E\left(h_t^2\right)<\infty$ and guaranteed by the condition of weak stationarity. \\
(ii) The requirements in Basawa et al. (1976) for the asymptotic normality of the maximum likelihood estimate are:\\
(1) $\frac{1}{\sqrt{T}}\sum_t\nabla l_t(\bdsm{\theta}_0)\stackrel{\mathscr{D}}{\to}N(\bdsm{0}, B)$, $T\to\infty$, for a nonrandom positive definite matrix $B$; \\
(2) -$\frac{1}{T}\sum_t\nabla^2l_t(\bdsm{\theta}_0)\stackrel{\mathcal{P}}{\to}A$, $T\to\infty$, for a nonrandom positive definite matrix $A$; \\
(3) Condition (3) for consistency. \\
For (1), since $E\left(\nabla l_t(\bdsm{\theta}_0)\right|\mathcal{F}_{t-1})=0$, by the martingale central limit theorem, 
\begin{eqnarray*}
  \frac{1}{\sqrt{T}}\sum_t\nabla l_t(\bdsm{\theta}_0)\stackrel{\mathscr{D}}{\to}N\left(\bdsm{0}, \bdsm{I}(\bdsm{\theta})\right),\ \ 
  T\to\infty,
\end{eqnarray*}
where $\bdsm{I}(\bdsm{\theta}_0)=E\left[\nabla\nabla^{T}l_t(\bdsm{\theta}_0)\right]$ is the Fisher information matrix. Condition (2) is already shown by the ergodic theorem in the preceding argument. Hence, the proof is completed. 

\end{proof}

\subsection{Proof of Corollary \ref{cor:stat}}
\begin{proof}
It is immediate from Theorem \ref{thm:mean}, \ref{thm:var}, and \ref{thm:cov}.
\end{proof}

\subsection{Proof of Corollary \ref{cor:acf}}
\begin{proof}
The Auto-correlation Function (ACF) of $\left\{ r_{t}\right\} $ is
defined to be $\rho(s)=\mbox{Corr}\left(r_{t},r_{t+s}\right)=\dfrac{\gamma(s)}{\gamma(0)}$.
Then the ACF of $\left\{ r_{t}\right\} $ is 
\[
\rho(s)=\begin{cases}
1, & s=0\\
\dfrac{kE\left(h_{t}h_{t+s}\eta_{t}\right)-k^{2}\left[E\left(h_{t}\right)\right]^{2}}{\left(1+k+k^{2}\right)E\left(h_{t}^{2}\right)-k^{2}\left[E\left(h_{t}\right)\right]^{2}}, & |s|>0.
\end{cases}
\]
\end{proof}

\subsection{Lemmas}\label{lemmas}
\begin{lem}\label{eta-X and h-h-eta}
\begin{eqnarray*}
  (i)\ &&E\left(\eta_{t}x_{t}\right)=\alpha_{1}\sqrt{\dfrac{2}{\pi}}\cdot k+\beta_{1}\left(k+k^{2}\right)+\gamma_{1}k;\\
 (ii)\ &&E\left(h_{t}h_{t+s}\eta_{t}\right)=\dfrac{\mu^{2}k}{C_{1}-1}\left(-\dfrac{C_{1}^{s}-1}{C_{1}-1}+\dfrac{C_{1}^{s}+C_{1}^{s-1}}{C_{2}-1}\cdot\left[\alpha_{1}\sqrt{\dfrac{2}{\pi}}+\beta_{1}\left(1+k\right)+\gamma_{1}\right]\right).
\end{eqnarray*}
 \end{lem} 

\begin{proof} 

(i)
\begin{eqnarray*}
E\left(\eta_{t}x_{t}\right) & = & E\left[\eta_{t}\cdot\left(\alpha_{1}\left|\varepsilon_{t}\right|+\beta_{1}\eta_{t}+\gamma_{1}\right)\right]\\
 & = & \alpha_{1}E\left(\left|\varepsilon_{t}\right|\right)\cdot E\left(\eta_{t}\right)+\beta_{1}\cdot E\left(\eta_{t}^{2}\right)+\gamma_{1}E\left(\eta_{t}\right)\\
 & = & \alpha_{1}\sqrt{\dfrac{2}{\pi}}\cdot k+\beta_{1}\left(k+k^{2}\right)+\gamma_{1}k.
\end{eqnarray*}

(ii) First, we expand $h_{t+s}$ recursively:
\begin{eqnarray*}
h_{t+s} & = & \mu+x_{t+s-1}h_{t+s-1}\\
 & = & \mu+\mu x_{t+s-1}+x_{t+s-1}x_{t+s-2}h_{t+s-2}\\
 & = & \cdots\\
 & = & \mu+\mu x_{t+s-1}+\mu x_{t+s-1}x_{t+s-2}+\cdots\\
 &  & +\mu x_{t+s-1}x_{t+s-2}\cdots x_{t+1}+x_{t+s-1}x_{t+s-2}\cdots x_{t}h_{t}\\
 & = & \mu\left(1+\sum_{i=1}^{s-1}\prod_{j=1}^{i}x_{t+s-j}\right)+h_{t}\prod_{j=1}^{s}x_{t+s-j}.
\end{eqnarray*}
Consequently,
\begin{eqnarray*}
h_{t}h_{t+s}\eta_{t} & = & h_{t}\eta_{t}\cdot\left(\mu\left(1+\sum_{i=1}^{s-1}\prod_{j=1}^{i}x_{t+s-j}\right)+h_{t}\prod_{j=1}^{s}x_{t+s-j}\right)\\
 & = & \mu h_{t}\eta_{t}+\mu h_{t}\eta_{t}\sum_{i=1}^{s-1}\prod_{j=1}^{i}x_{t+s-j}+h_{t}^{2}\eta_{t}\prod_{j=1}^{s}x_{t+s-j}.
\end{eqnarray*}

Then the expected value is found to be
\begin{eqnarray*}
E\left(h_{t}h_{t+s}\eta_{t}\right) & = & \mu\cdot E\left(h_{t}\right)\cdot E\left(\eta_{t}\right)+\mu\cdot E\left(h_{t}\right)\cdot E\left(\eta_{t}\right)\cdot E\left[\sum_{i=1}^{s-1}\prod_{j=1}^{i}\left(x_{t+s-j}\right)\right]\\
 &  & +E\left(h_{t}^{2}\right)\cdot E\left(\eta_{t}\prod_{j=1}^{s}x_{t+s-j}\right)\\
 & = & \mu k\cdot E\left(h_{t}\right)+\mu k\cdot E\left(h_{t}\right)\cdot\left[\sum_{i=1}^{s-1}\prod_{j=1}^{i}E\left(x_{t+s-j}\right)\right]\\
 &  & +E\left(h_{t}^{2}\right)\cdot E\left(\eta_{t}x_{t}\right)\cdot\prod_{j=1}^{s-1}E\left(x_{t+s-j}\right)\\
 & = & \mu k\cdot E\left(h_{t}\right)+\mu k\cdot E\left(h_{t}\right)\cdot\dfrac{C_{1}\left(C_{1}^{s-1}-1\right)}{C_{1}-1}\\
 &  & +E\left(h_{t}^{2}\right)\cdot E\left(\eta_{t}x_{t}\right)\cdot C_{1}^{s-1},
\end{eqnarray*}
where $C_{1}=E\left(x_{t}\right)$.\\

Finally, remembering that $E\left(h_{t}^{2}\right)=\mu^{2}\dfrac{C_{1}+1}{\left(C_{2}-1\right)\left(C_{1}-1\right)}$, the above calculation is simplified to
\begin{eqnarray*}
E\left(h_{t}h_{t+s}\eta_{t}\right) & = & \mu k\cdot\dfrac{\mu}{1-C_{1}}\cdot\left[1+\dfrac{C_{1}\left(C_{1}^{s-1}-1\right)}{C_{1}-1}\right]\\
 &  & +E\left(h_{t}^{2}\right)\cdot\left[\alpha_{1}\sqrt{\dfrac{2}{\pi}}\cdot k+\beta_{1}\left(k+k^{2}\right)+\gamma_{1}k\right]\cdot C_{1}^{s-1}\\
 & = & \mu^{2}k\cdot\dfrac{1}{1-C_{1}}\cdot\left[\dfrac{C_{1}-1+C_{1}\left(C_{1}^{s-1}-1\right)}{C_{1}-1}\right]\\
 &  & +E\left(h_{t}^{2}\right)\cdot\left[\alpha_{1}\sqrt{\dfrac{2}{\pi}}\cdot k+\beta_{1}\left(k+k^{2}\right)+\gamma_{1}k\right]\cdot C_{1}^{s-1}\\
 & = & -\mu^{2}k\cdot\dfrac{C_{1}^{s}-1}{\left(C_{1}-1\right)^{2}}+\mu^{2}\dfrac{C_{1}+1}{\left(C_{2}-1\right)\left(C_{1}-1\right)}\cdot\left[\alpha_{1}\sqrt{\dfrac{2}{\pi}}\cdot k+\beta_{1}\left(k+k^{2}\right)+\gamma_{1}k\right]\cdot C_{1}^{s-1}\\
 & = & \dfrac{\mu^{2}k}{C_{1}-1}\left(-\dfrac{C_{1}^{s}-1}{C_{1}-1}+\dfrac{C_{1}^{s}+C_{1}^{s-1}}{C_{2}-1}\cdot\left[\alpha_{1}\sqrt{\dfrac{2}{\pi}}+\beta_{1}\left(1+k\right)+\gamma_{1}\right]\right).
\end{eqnarray*}

\end{proof}

\begin{lem}\label{expected value of x_t ^2}
\begin{eqnarray*}
E\left(x_{t}^{2}\right) & = & \alpha_{1}^{2}+\beta_{1}^{2}\left(k+k^{2}\right)+\gamma_{1}^{2}+2\alpha_{1}\beta_{1}\sqrt{\dfrac{2}{\pi}}k+2\alpha_{1}\gamma_{1}\sqrt{\dfrac{2}{\pi}}+2\beta_{1}\gamma_{1}k.
\end{eqnarray*}
\end{lem}
\begin{proof}
\begin{eqnarray*}
E\left(x_{t}^{2}\right) & = & E\left(\left(\alpha_{1}\left|\varepsilon_{t}\right|+\beta_{1}\eta_{t}+\gamma_{1}\right)^{2}\right)\\
 & = & E\left(\alpha_{1}^{2}\varepsilon_{t}^{2}+\beta_{1}^{2}\eta_{t}^{2}+\gamma_{1}^{2}+2\alpha_{1}\beta_{1}\left|\varepsilon_{t}\right|\eta_{t}+2\alpha_{1}\gamma_{1}\left|\varepsilon_{t}\right|+2\beta_{1}\gamma_{1}\eta_{t}\right)\\
 & = & \alpha_{1}^{2}E\left(\varepsilon_{t}^{2}\right)+\beta_{1}^{2}E\left(\eta_{t}^{2}\right)+\gamma_{1}^{2}+2\alpha_{1}\beta_{1}E\left(\left|\varepsilon_{t}\right|\right)\cdot E\left(\eta_{t}\right)+2\alpha_{1}\gamma_{1}E\left(\left|\varepsilon_{t}\right|\right)+2\beta_{1}\gamma_{1}E\left(\eta_{t}\right)\\
 & = & \alpha_{1}^{2}+\beta_{1}^{2}\left(k+k^{2}\right)+\gamma_{1}^{2}+2\alpha_{1}\beta_{1}\sqrt{\dfrac{2}{\pi}}k+2\alpha_{1}\gamma_{1}\sqrt{\dfrac{2}{\pi}}+2\beta_{1}\gamma_{1}k.
\end{eqnarray*}
\end{proof}

\begin{lem}\label{stationarity-ht}
Consider the Int-GARCH(1,1,1) model and assume the condition for mean stationarity in Theorem \ref{thm:mean}. Then, the process $\left\{h_t\right\}$ is strictly stationary and ergodic. Assume in addition that $E\left(x_t^p\right)<1$, $p>0$, then $E\left(h_t^p\right)<\infty$. 
\begin{proof}
The process $\left\{h_t\right\}$ is defined by the stochastic recurrence equation (\ref{def_x}), where $x_t=\alpha_1|\epsilon_{t-1}|+\beta\eta_{t-1}+\gamma_1$, $t\in\mathbb{N}$ is an i.i.d. process whose moments are all finite. Therefore, $\left\{h_t\right\}$ is strictly stationary and ergodic if 
\begin{equation}\label{lem3:eqn-1}
  E\left[\log\left(x_t\right)\right]<0.
\end{equation}
See, e.g., Nelson (1990). Since $E\left[\log\left(x_t\right)\right]\leq\log\left[E\left(x_t\right)\right]$, (\ref{lem3:eqn-1}) is automatically true under the assumption of mean stationarity, i.e., $E\left(x_t\right)<1$. Furthermore, $E\left(h_t^p\right)<\infty$ iff $E\left(x_t^p\right)<1$, $p>0$ (Nelson 1990, Theorem 3). This finishes the proof. 
\end{proof}
\end{lem}

\end{document}